\documentclass[12pt,a4paper]{article}

\usepackage{amsmath,amssymb}
\usepackage{mathtools}
\usepackage{graphicx}
\usepackage[nosort]{cite}

\usepackage{tikz}
\usetikzlibrary{decorations.markings}
\usepackage{epsfig}  
\usepackage{enumerate}
\usepackage{relsize}
\usepackage{dsfont}
\usetikzlibrary{calc} 
\usetikzlibrary{patterns} 

\newcommand{\mathsym}[1]{{}}

\usepackage{amsthm,amsbsy,color,multicol}
\usepackage{array,calc}
\usepackage{rotating}
\usepackage{a4wide,bm}
\usepackage{bbm}
\usepackage[a4paper,text={450pt,650pt},centering]{geometry}
\usepackage{setspace}

\let\oldbfseries=\bfseries
\let\oldmdseries=\mdseries
\let\oldnormalfont=\normalfont
\renewcommand{\bfseries}{\oldbfseries\boldmath}
\renewcommand{\mdseries}{\oldmdseries\unboldmath}
\renewcommand{\normalfont}{\oldnormalfont\unboldmath}

\allowdisplaybreaks[3]

\numberwithin{equation}{section}

\usepackage[font=small,labelfont=bf,width=0.85\textwidth]{caption}

\ifx\hypersetup\sadfkjashdfkxja\newcommand\hypersetup[1]{}\fi

\hypersetup{plainpages=false}
\hypersetup{pdfpagemode=UseNone}
\hypersetup{bookmarksnumbered=true}
\hypersetup{pdfstartview=FitH}
\hypersetup{colorlinks=false}
\hypersetup{citebordercolor={.5 1 .5}}
\hypersetup{urlbordercolor={.5 1 1}}
\hypersetup{linkbordercolor={1 .7 .7}}

\DeclareMathSymbol{\Gamma}{\mathalpha}{letters}{"00}
\DeclareMathSymbol{\Delta}{\mathalpha}{letters}{"01}
\DeclareMathSymbol{\Theta}{\mathalpha}{letters}{"02}
\DeclareMathSymbol{\Lambda}{\mathalpha}{letters}{"03}
\DeclareMathSymbol{\Xi}{\mathalpha}{letters}{"04}
\DeclareMathSymbol{\Pi}{\mathalpha}{letters}{"05}
\DeclareMathSymbol{\Sigma}{\mathalpha}{letters}{"06}
\DeclareMathSymbol{\Upsilon}{\mathalpha}{letters}{"07}
\DeclareMathSymbol{\Phi}{\mathalpha}{letters}{"08}
\DeclareMathSymbol{\Psi}{\mathalpha}{letters}{"09}
\DeclareMathSymbol{\Omega}{\mathalpha}{letters}{"0A}

\newcommand{\gen}[1]{\mathrm{#1}}

\newcommand{\dd}{\mathrm{d}}
\newcommand{\ii}{\mathrm{i}}

\newcommand*\widebar[1]{%
  \hbox{%
    \vbox{%
      \hrule height 0.5pt 
      \kern0.25ex
      \hbox{%
        \kern-0.3em
        \ensuremath{#1}%
        \kern-0.1em
      }%
    }%
  }%
}

\ifx\genfrac\sdflkaj\else\fi

\newcommand{\ket}[1]{\left|#1\right\rangle}      
\newcommand{\bra}[1]{\left\langle #1\right|}     

\newcommand{\alg}[1]{\mathfrak{#1}}

\newcommand{\beq}{\begin{equation}}
\newcommand{\eeq}{\end{equation}}

\def\[{\begin{equation}}
\def\]{\end{equation}}
\def\<{\begin{eqnarray}}
\def\>{\end{eqnarray}}

\newtheorem{mydef}{Definition}

\newtheorem{theorem}{Theorem}
\newtheorem{lemma}{Lemma} 
\newtheorem{remark}{Remark}

\makeatletter
\def\mr@ignsp#1 {\ifx\:#1\@empty\else #1\expandafter\mr@ignsp\fi}%
\newcommand{\multiref}[1]{\begingroup
\xdef\mr@no@sparg{\expandafter\mr@ignsp#1 \: }%
\def\mr@comma{}%
\@for\mr@refs:=\mr@no@sparg\do{\mr@comma\def\mr@comma{,}\ref{\mr@refs}}%
\endgroup}
\makeatother

\newcommand{\hypref}[2]{\ifx\href\asklfhas #2\else\href{#1}{#2}\fi}
\newcommand{\Secref}[1]{Section~\multiref{#1}}

\newcommand{\Appref}[1]{Appendix~\multiref{#1}}

\newcommand{\Figref}[1]{Figure~\multiref{#1}}

\renewcommand{\eqref}[1]{(\multiref{#1})}

\makeatletter
\newlength{\apb@width}
\newcommand{\autoparbox}[2][c]{\settowidth{\apb@width}{#2}\parbox[#1]{\apb@width}{#2}}

\makeatother

\ifx\href\asklfhas\newcommand{\href}[2]{#2}\fi

\begin{document}

\renewcommand{\thefootnote}{\fnsymbol{footnote}}
\thispagestyle{empty}
\begin{flushright}\footnotesize
ITP-UU-14/14 \\
SPIN-14/12
\end{flushright}
\vspace{1cm}

\begin{center}%
{\Large\bfseries%
\hypersetup{pdftitle={Reflection algebra and functional equations}}%
Reflection algebra and \\ functional equations%
\par} \vspace{2cm}%

\textsc{W. Galleas and J. Lamers}\vspace{5mm}%
\hypersetup{pdfauthor={Wellington Galleas}}%

\textit{Institute for Theoretical Physics and Spinoza Institute, \\ Utrecht University, Leuvenlaan 4,
3584 CE Utrecht, \\ The Netherlands}\vspace{3mm}%

\verb+w.galleas@uu.nl+, %
\verb+j.lamers@uu.nl+

\par\vspace{3cm}

\textbf{Abstract}\vspace{7mm}

\begin{minipage}{12.7cm}
In this work we investigate the possibility of using the reflection algebra as a source of functional equations. 
More precisely, we obtain functional relations determining the partition function of the six-vertex model with  
domain-wall boundary conditions and one reflecting end. The model's partition function is expressed as a multiple-contour integral that
allows the homogeneous limit to be obtained straightforwardly. Our functional equations are also shown to give rise to a
consistent set of partial differential equations satisfied by the partition function.

\hypersetup{pdfkeywords={Functional equations, domain walls, open boundaries}}%
\hypersetup{pdfsubject={}}%
\end{minipage}
\vskip 2cm
{\small July 2014}

\end{center}

\newpage
\renewcommand{\thefootnote}{\arabic{footnote}}
\setcounter{footnote}{0}

\tableofcontents

\section{Introduction}
\label{sec:intro}

The study of correlation functions of quantum-integrable systems is intrinsically related
to partition functions of vertex models with special boundary conditions. In particular, the 
case of domain-wall boundaries is of fundamental importance and this fact was first noticed for the 
one-dimensional Heisenberg chain with periodic boundary conditions and the associated six-vertex model
\cite{Korepin_1982}. Among the results of \cite{Korepin_1982} there is a recurrence relation determining the partition function
of the six-vertex model with domain-wall boundary conditions, which was later on solved by Izergin in terms of a 
determinant \cite{Izergin_1987}.
Subsequently, a determinant representation for scalar products of Bethe vectors under certain specializations of
the parameters (so called on-shell scalar products) was obtained by Slavnov \cite{Slavnov_1989}. 
Those results represented the first steps towards the computation of exact correlation functions of quantum-integrable systems,
 but it is worth remarking that the problem of computing norms of Bethe wave functions was first
considered by Gaudin in the case of the non-linear Schr\"odinger model \cite{Gaudin_book}.
Gaudin also formulated the hypothesis that such norms would be given by certain Jacobian determinants.
This hypothesis was subsequently proved by Korepin \cite{Korepin_book} even for more general models. 

Scalar products of Bethe vectors are the building blocks of correlation functions, and having them expressed as
determinants for the Heisenberg chain paved the way for the use of well-established techniques.  
For instance, the aforementioned determinant formulae in the thermodynamic limit become 
determinants of integral operators, i.e.~Fredholm determinants, allowing the derivation of differential equations describing 
correlation functions \cite{Its_1990}. The asymptotic behavior of correlation functions was then derived from the analysis of 
the respective differential equations. Interestingly, the language of differential equations seems to be quite appropriate for
the description of correlation functions. A remarkable example of this is provided by the Knizhnik-Zamolodchikov equation describing
multi-point correlation functions of primary fields in conformal field theories \cite{Knizhnik_1984}.  

As far as the one-dimensional Bose gas is concerned, the results of \cite{Its_1990} are largely due to the existence of particular determinant 
representations for correlation functions. However, the recent works \cite{Belliard1, Belliard2, Belliard3, Belliard4} have raised
doubts on the existence of such determinant representations for models based on higher-rank symmetry algebras.
Nevertheless, it is still important to remark here that several advances in the computation of correlation functions have been obtained
through the algebraic Bethe ansatz in \cite{Kitanine_1999, Kitanine_2000, Kitanine_2002, Kitanine_2005, Terras_2013a, Terras_2013b}
and through Sklyanin's separation of variables in \cite{Niccoli_2013a, Niccoli_2013b}.

Integrable spin chains can also be addressed directly in the thermodynamic limit through the vertex-operator approach
due to the quantum affine symmetry exhibited in that limit \cite{Jimbo_1993, Kedem_1995a}.
Alternatively, one can also consider the $q$-Onsager approach described in \cite{Baseilhac_2013}.
Within the vertex-operator approach the description of correlation functions is done by
means of the quantized Knizhnik-Zamolodchikov equation \cite{Jimbo_1992, Jimbo_1993a, Jimbo_1996}.
The latter is not a differential equation but rather a functional equation describing the matrix coefficients of a product of 
intertwining operators for a given affine Lie algebra \cite{Reshetikhin_1992}. Interestingly, the partition function of the six-vertex
model with domain-wall boundaries, which initiated the study of correlation functions for one-dimensional quantum spin chains, can also
be described by a functional equation resembling certain aspects of the Knizhnik-Zamolodchikov equation \cite{Galleas_proc}. 
Moreover, the functional equation for that partition function can further be translated into a partial differential equation, as
was shown in \cite{Galleas_2011, Galleas_proc}.

The derivation of such functional equations is based on an algebraic-functional approach initially proposed
for spectral problems in \cite{Galleas_2008} and extended for correlations in the series of works 
\cite{Galleas_2010,Galleas_2011,Galleas_2012, Galleas_2013, Galleas_SCP}. In particular, those equations allow for
the derivation of integral representations for partition functions with domain-wall boundaries \cite{Galleas_2012, Galleas_2013}
and scalar products of Bethe vectors \cite{Galleas_SCP}. A fundamental ingredient within this algebraic-functional
approach is the Yang-Baxter algebra, which is a common algebraic structure underlying quantum-integrable
systems. On the other hand, integrable systems with open boundary conditions are governed by the reflection
algebra \cite{Sklyanin_1988} and in the present paper we show that this algebra can also be exploited along the 
lines of \cite{Galleas_proc}.

Correlation functions for the Heisenberg chain with open boundary conditions have been studied
in the literature through a variety of approaches \cite{Kitanine_2007, Kitanine_2008, Kedem_1995b}. In particular, 
a partition function with domain-wall boundaries has also been defined in that case \cite{Tsuchiya_1998}. 
As a matter of fact, the partition function introduced in \cite{Tsuchiya_1998} considers both domain-wall boundaries and
one reflecting end, and it can also be expressed as a determinant along the lines of \cite{Izergin_1987}.
Moreover, this partition function has also found interesting applications in the computation of physical properties of
the $XXZ$ spin chain at finite temperature. For instance, the surface free energy of the $XXZ$ spin chain has been expressed
in \cite{Goehmann_2005} as the expectation value of a product of projection operators.
This expectation value was then demonstrated in \cite{Kozlowski_2012} to be precisely the partition function
of the six-vertex model with one reflecting end and domain-wall boundaries described in \cite{Tsuchiya_1998}.
Here we shall study this partition function through the algebraic-functional method
developed in \cite{Galleas_2012, Galleas_2013}. This approach will not only allow us to find a new representation
for the model's partition function, but it will also unveil a set of partial differential equations satisfied by the latter.

This paper is organized as follows. In \Secref{sec:CONV} we introduce definitions and conventions which
shall be employed throughout this work. In \Secref{sec:AF} we describe the algebraic-functional approach
in terms of the reflection algebra, and use this method to derive functional equations characterizing
the partition function of the six-vertex model with one reflecting end and domain-wall boundaries. The solution of our equation
is then given in \Secref{sec:AF}. We proceed with the analysis of our functional equation in \Secref{sec:PDE}
where we extract a set of partial differential equations satisfied by the model's partition function. 
Concluding remarks are discussed in \Secref{sec:conclusion} while proofs and technical details are left
for the Appendices.

\section{Definitions and conventions}
\label{sec:CONV}

The most well studied cases of integrable lattice systems are those defined on a finite interval with periodic boundary conditions,
although more general types of boundaries can also be considered. Among those systems, a prominent class is formed by one-dimensional
spin chains with open ends and particular terms characterizing the reflection at the boundaries.
Those models can also be solved by means of the Bethe ansatz and the first results
on that direction have been obtained in \cite{Gaudin_1971, Alcaraz_1987}. Boundary terms are normally
not expected to modify the infinite-volume properties of physical systems, although  
counter\-examples for that common belief have been reported in the literature \cite{Korepin_Justin_2000}. 
Nevertheless, even for the cases where infinite-volume properties remain the same, boundary terms can still change the finite-size
corrections of massless systems defined on a strip of width $L$ \cite{Alcaraz_1987}. 
The latter is able to provide fundamental information concerning the underlying conformal field theory as shown in \cite{Cardy_1986}. 

The study of lattice systems with open boundary conditions gained a large impulse with the formulation of
the Quantum Inverse Scattering Method (QISM) for that class of models \cite{Sklyanin_1988}. The approach developed in \cite{Sklyanin_1988}
is based on Cherednik's condition for factorised scattering with reflection \cite{Cherednik_1984}, although it can also
be regarded in the context of vertex models of Statistical Mechanics \cite{Baxter_book}. The latter is the perspective
to be adopted here, and in what follows we shall briefly describe the ingredients required for the construction of the six-vertex
model with domain walls and one reflecting end as defined in \cite{Tsuchiya_1998}.

\paragraph{The $\mathcal{U}_q [ \widehat{\alg{sl}}(2) ]$ $\mathcal{R}$-matrix.} Integrable vertex models are  
characterized by an $\mathcal{R}$-matrix $\mathcal{R}: \; \mathbb{C} \rightarrow \gen{End}( \mathbb{V} \otimes \mathbb{V} )$
satisfying the Yang-Baxter equation, namely
\<
\label{ybe}
\mathcal{R}_{12}(\lambda_1 - \lambda_2) \mathcal{R}_{13}(\lambda_1) \mathcal{R}_{23}(\lambda_2) =  \mathcal{R}_{23}(\lambda_2) \mathcal{R}_{13}(\lambda_1) \mathcal{R}_{12}(\lambda_1 - \lambda_2) \; . 
\>
In Eq. (\ref{ybe}) we are using the standard notation $\mathcal{R}_{ij} \in \gen{End}(\mathbb{V}_i \otimes \mathbb{V}_j)$ 
and $\lambda_i$ denote arbitrary complex parameters. For the six-vertex model we have $\mathbb{V}_i  \coloneqq  \mathbb{V} \cong \mathbb{C}^2$
and
\<
\label{rmat}
\mathcal{R} (\lambda) = \left( \begin{matrix}
a(\lambda) & 0 & 0 & 0 \\
0 & b(\lambda) & c(\lambda) & 0 \\
0 & c(\lambda) & b(\lambda) & 0 \\
0 & 0 & 0 & a(\lambda) \end{matrix} \right)
\>
where $a(\lambda)=\sinh{(\lambda + \gamma)}$, $b(\lambda)=\sinh{(\lambda)}$ and $c(\lambda)=\sinh{(\gamma)}$ with 
anisotropy parameter $\gamma \in \mathbb{C}$.

\paragraph{Monodromy matrices.} Let $\mathbb{V}_{0} \cong \mathbb{C}^2$ and 
$\mathbb{V}_{\mathcal{Q}}  \coloneqq   (\mathbb{C}^2)^{\otimes L}$ for lattice length $L \in \mathbb{Z}_{> 0}$. Also let
$\lambda$ and $\mu_j$ ($1 \leq j \leq L$) be arbitrary complex parameters. Then we consider operators
$\tau, \bar{\tau} \colon \mathbb{C} \to \gen{End}( \mathbb{V}_{0} \otimes \mathbb{V}_{\mathcal{Q}})$ defined as  
the following ordered products:
\<
\label{mono}
\tau (\lambda) \coloneqq \mathop{\overleftarrow\prod}\limits_{1 \le j \le L } \mathcal{R}_{0 j} (\lambda - \mu_j )  \qquad \mbox{and} \qquad 
\bar{\tau} (\lambda) \coloneqq \mathop{\overrightarrow\prod}\limits_{1 \le j \le L } \mathcal{R}_{0 j} (\lambda + \mu_j ) \; .
\>
The operators $\tau$ and $\bar{\tau}$ are usually referred to as \emph{monodromy matrices} while 
$\mathbb{V}_{0}$ and $\mathbb{V}_{\mathcal{Q}}$ are called \emph{auxiliary} and \emph{quantum spaces} respectively.
Since $\mathbb{V}_{0} \cong \mathbb{C}^2$, the monodromy matrices (\ref{mono}) can be recasted as
\[
\label{abcd}
\tau (\lambda) = \left( \begin{matrix}
A(\lambda) & B(\lambda) \\
C(\lambda) & D(\lambda) \end{matrix} \right) \qquad \mbox{and} \qquad 
\bar{\tau}(\lambda) = \left( \begin{matrix}
\bar{A}(\lambda) & \bar{B}(\lambda) \\
\bar{C}(\lambda) & \bar{D}(\lambda) \end{matrix} \right) \; ,
\]
with entries in $\gen{End}(\mathbb{V}_{\mathcal{Q}})$.

\paragraph{Yang-Baxter algebra.} Let $\mathcal{R}$ be given by (\ref{rmat}) and consider $\gen{U} \in \{ \tau , \bar{\tau} \}$.
Then the following quadratic algebra is fulfilled by $\gen{U}$, 
\[
\label{yba}
\mathcal{R}_{12} (\lambda_1 - \lambda_2) \gen{U}_1 (\lambda_1) \gen{U}_2 (\lambda_2) = \gen{U}_2 (\lambda_2) \gen{U}_1 (\lambda_1) \mathcal{R}_{12} (\lambda_1 - \lambda_2) \; ,
\]
due to the Yang-Baxter equation (\ref{ybe}). Here we are employing the notation $\gen{Y}_1 \coloneqq \gen{Y} \otimes \gen{id}_2$ and
$\gen{Y}_2 \coloneqq \gen{id}_1 \otimes \gen{Y}$ for any operator $\gen{Y} \in \mbox{End}(\mathbb{V})$, where the symbol $\gen{id}_j$ 
stands for the identity operator in $\mbox{End}(\mathbb{V}_j)$. The relation (\ref{yba}) is commonly referred to 
as \emph{Yang-Baxter algebra} and it describes commutation relations for the entries of the monodromy matrices (\ref{abcd}). 

\paragraph{Reflection equation.} Within the framework of the QISM developed in \cite{Sklyanin_1988}, integrable boundary conditions are
characterized by a matrix $\mathcal{K}$ satisfying the so called reflection equation. This equation, which was first proposed
in \cite{Cherednik_1984} in the context of factorised scattering, reads
\<
\label{REQ}
&& \mathcal{R}_{12} (\lambda_1 - \lambda_2) \mathcal{K}_1 (\lambda_1) \mathcal{R}_{12} (\lambda_1 + \lambda_2) \mathcal{K}_2 (\lambda_2) \nonumber \\
&& = \mathcal{K}_2 (\lambda_2) \mathcal{R}_{12} (\lambda_1 + \lambda_2) \mathcal{K}_1 (\lambda_1) \mathcal{R}_{12} (\lambda_1 - \lambda_2) \; .
\>
Here we shall restrict ourselves to a particular solution of (\ref{REQ}) associated to the $\mathcal{R}$-matrix (\ref{rmat}).
This solution is explicitly given by
\[
\label{kmat}
\mathcal{K}(\lambda) = \left( \begin{matrix}
\kappa_{+}(\lambda) & 0 \\
0 & \kappa_{-}(\lambda)  \end{matrix} \right) \; ,
\]
where $\kappa_{\pm}(\lambda) = \sinh{(h \pm \lambda)}$ and $h \in \mathbb{C}$ is an arbitrary parameter 
describing the interaction at one of the boundaries.

\begin{remark} Within the context of one-dimensional integrable spin chains, the $\mathcal{K}$-matrix (\ref{kmat}) describes the
reflection only at one of the ends of an open system. If we would want to describe reflection at the opposite end we would also need to
introduce a matrix $\bar{\mathcal{K}}$ satisfying an independent equation isomorphic to (\ref{REQ}). More details on this construction can
be found in \cite{Sklyanin_1988}. 
\end{remark}

\paragraph{Reflection algebra.} Let the operator $\mathcal{T} \colon \mathbb{C} \to \gen{End}( \mathbb{V}_{0} \otimes \mathbb{V}_{\mathcal{Q}})$
be defined as 
\[
\label{full_mono}
\mathcal{T}(\lambda) \coloneqq \tau(\lambda) \mathcal{K}(\lambda) \bar{\tau}(\lambda) \; .
\]
Then, following \cite{Sklyanin_1988}, one can show that $\mathcal{T}$ satisfies the following
quadratic algebra,
\<
\label{REA}
&& \mathcal{R}_{12} (\lambda_1 - \lambda_2) \mathcal{T}_1 (\lambda_1) \mathcal{R}_{12} (\lambda_1 + \lambda_2) \mathcal{T}_2 (\lambda_2) \nonumber \\
&& = \mathcal{T}_2 (\lambda_2) \mathcal{R}_{12} (\lambda_1 + \lambda_2) \mathcal{T}_1 (\lambda_1) \mathcal{R}_{12} (\lambda_1 - \lambda_2) \; .
\>
The relation (\ref{REA}) will be referred to as \emph{reflection algebra} and it follows from the reflection equation obeyed by the matrix 
$\mathcal{K}$ together with properties satisfied by $\tau$ and $\bar{\tau}$. The operator $\mathcal{T}$ is commonly referred to
as \emph{double-row monodromy matrix} and, similarly to (\ref{abcd}), it can be recasted as
\[
\label{ABCD}
\mathcal{T}(\lambda) = \left( \begin{matrix}
\mathcal{A}(\lambda) & \mathcal{B}(\lambda) \\
\mathcal{C}(\lambda) & \mathcal{D}(\lambda) \end{matrix} \right) \; .
\]
In this way (\ref{REA}) encodes commutation relations for the operators $\mathcal{A}, \mathcal{B}, \mathcal{C} , \mathcal{D} \in \mbox{End}(\mathbb{V}_{\mathcal{Q}})$.

\paragraph{Highest/lowest weight vectors.} The vectors $\ket{0}, \ket{\bar{0}} \in \mathbb{V}_{\mathcal{Q}}$
defined as
\[
\label{zero}
\ket{0} \coloneqq \left( \begin{matrix} 1 \\ 0 \end{matrix} \right)^{\otimes L} 
\qquad \mbox{and} \qquad 
\ket{\bar{0}} \coloneqq \left( \begin{matrix} 0 \\ 1 \end{matrix} \right)^{\otimes L} \; 
\]
are respectively $\alg{sl}(2)$ highest- and lowest-weight vectors. Due to the structure of (\ref{rmat})
we can easily compute the action of the entries of (\ref{ABCD}) on the vectors (\ref{zero}). 
This computation can be found in \Appref{sec:ZERO}. It turns out that $\mathcal{A}(\lambda) \ket{0} = \Lambda_{\mathcal{A}} (\lambda) \ket{0}$, 
$\tilde{\mathcal{D}}(\lambda) \ket{0} = \Lambda_{\tilde{\mathcal{D}}} (\lambda) \ket{0}$ and
$\bra{\bar{0}} \mathcal{A}(\lambda)  = \bar{\Lambda}_{\mathcal{A}} (\lambda) \bra{\bar{0}}$,
where we have defined the operator $\tilde{\mathcal{D}}(\lambda) \coloneqq \mathcal{D}(\lambda)  - \frac{c(2 \lambda)}{a(2 \lambda)} \mathcal{A}(\lambda)$
for later convenience. The functions $\Lambda_{\mathcal{A}}$, $\Lambda_{\tilde{\mathcal{D}}}$ and $\bar{\Lambda}_{\mathcal{A}}$
explicitly read
\begin{eqnarray}
\label{lambda}
\Lambda_{\mathcal{A}} (\lambda) & \coloneqq & b(h + \lambda) \prod_{j=1}^{L} a(\lambda - \mu_j) a(\lambda + \mu_j)  \nonumber \\
\Lambda_{\tilde{\mathcal{D}}} (\lambda) &\coloneqq & - \frac{b(2 \lambda)}{a(2 \lambda)} a(\lambda - h) \prod_{j=1}^{L} b(\lambda - \mu_j) b(\lambda + \mu_j)  \nonumber \\
\bar{\Lambda}_{\mathcal{A}} (\lambda) &\coloneqq & \frac{c(2 \lambda)}{a(2 \lambda)} b(h - \lambda) \prod_{j=1}^{L} a(\lambda - \mu_j) a(\lambda + \mu_j) \nonumber \\
&&+ \frac{b(2 \lambda)}{a(2 \lambda)} a(\lambda + h) \prod_{j=1}^{L} b(\lambda - \mu_j) b(\lambda + \mu_j) \; . 
\end{eqnarray}

\paragraph{Partition function.} Following the work \cite{Tsuchiya_1998}, the partition function of the six-vertex
model with one reflecting end and domain-wall boundaries is given by
\[
\label{PF}
\mathcal{Z}(\lambda_1 , \lambda_2 , \dots , \lambda_L) = \bra{\bar{0}} \mathop{\overrightarrow\prod}\limits_{1 \le j \le L } \mathcal{B}(\lambda_j) \ket{0} \; .
\]
This partition function is a multivariate function depending on $L$ spectral parameters $\lambda_j$, $L$ inhomogeneity parameters
$\mu_j$, the anisotropy parameter $\gamma$ and the boundary parameter $h$.
In \Secref{sec:AF} we shall describe how the reflection algebra (\ref{REA}) can be exploited in order to derive functional equations determining the partition function (\ref{PF}).

\paragraph{Diagrammatic representation.} The lattice system described by the partition function (\ref{PF}) can be more intuitively
depicted in terms of diagrams representing  the action of the Yang-Baxter and reflection algebras elements. For that it is 
convenient to write 
$\mathcal{R} = \mathcal{R}_{\alpha \beta}^{\alpha' \beta'} \; e_{\alpha \alpha'} \otimes e_{\beta  \beta'}$,
$\mathcal{K} = \mathcal{K}_{\alpha}^{\alpha'} \; e_{\alpha \alpha'}$ and 
$\mathcal{T} = \mathcal{T}_{\alpha}^{\alpha'} \; e_{\alpha \alpha'}$, where summation over repeated indices is assumed.
Here $\alpha , \alpha' , \beta , \beta' \in \{ 1 , 2 \}$ label the basis vectors of $\mathbb{V} \cong \mathbb{C}^2$, 
while $e_{\alpha \beta}$ is the matrix with entries $(e_{\alpha \beta})_{ij} = \delta_{\alpha  i} \delta_{\beta  j}$. 
The diagrammatic representation of $\mathcal{R}$ and $\mathcal{K}$ is given in \Figref{fig:RandK} while $\mathcal{T}$ is depicted in \Figref{fig:mono}. 
Using these conventions the partition function (\ref{PF}) is illustrated in \Figref{fig:partition} with external indices assuming the 
domain-wall configurations $\alpha_j , \beta_j = 1$ and $\alpha_j' , \beta_j' = 2$ for all $j$.
\begin{figure}
\begin{center}
\begin{tikzpicture}[font=\scriptsize,thick]
	\draw (-1.5,1) node[font=\normalsize]{$\mathcal{R}^{\alpha'\beta'}_{\alpha\beta} \ = $};
	\draw (0,1) node[left]{$\alpha$} -- (2,1) node[right]{$\alpha'$};
	\draw (1,0) node[below]{$\beta$} -- (1,2) node[above]{$\beta'$};

\begin{scope}[xshift=8.5cm, yshift=1cm]
	\draw (-1.2,0) node[font=\normalsize]{$\mathcal{K}^{\alpha'}_{\alpha} \ = $};
	\draw[rounded corners=6pt] (0,0) -- (-35:1.2) node[right]{$\alpha$};
	\draw[rounded corners=6pt] (0,0) -- (35:1.2) node[right]{$\alpha'$};
	\fill[preaction={fill,white},pattern=north east lines, pattern color=gray] (0,-1) rectangle (-.15,1) ; \draw (0,-1) -- (0,1);
\end{scope}
\end{tikzpicture}
\end{center}
\caption{Diagrammatic representation of the $\mathcal{R}$- and $\mathcal{K}$-matrices.}
\label{fig:RandK}
\end{figure}
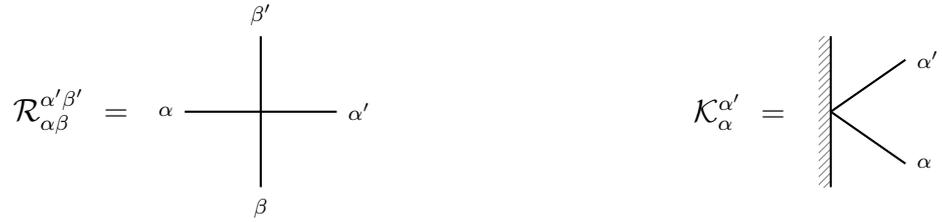
  
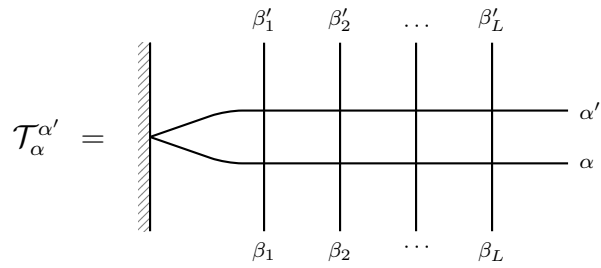
\begin{figure}
\begin{center}
\begin{tikzpicture}[font=\scriptsize,thick]
	\draw (-1.2,1.25) node[font=\normalsize]{$\mathcal{T}^{\alpha'}_{\alpha} \ = $};
	\draw[rounded corners=6pt] (0,1.25) -- (1,1.25-.35) -- (5.5,1.25-.35) node[right]{$\alpha$};
	\draw[rounded corners=6pt] (0,1.25) -- (1,1.25+.35) -- (5.5,1.25+.35) node[right]{$\alpha'$};
	\draw (1+.5,0) node[below]{$\beta_1$} -- (1+.5,2.5) node[above]{$\beta_1'$};
	\draw (2+.5,0) node[below]{$\beta_2$} -- (2+.5,2.5) node[above]{$\beta_2'$};
	\draw (3+.5,0) node[below,xshift=1pt]{$\cdots$} -- (3+.5,2.5) node[above,xshift=1pt]{$\cdots$};
	\draw (4+.5,0) node[below]{$\beta_L$} -- (4+.5,2.5) node[above]{$\beta_L'$};
	\fill[preaction={fill,white},pattern=north east lines, pattern color=gray] (0,0) rectangle (-.15,2.5) ; \draw (0,0) -- (0,2.5);
\end{tikzpicture}
\end{center}
\caption{The double-row monodromy matrix $\mathcal{T}$ depicted diagrammatically.}
\label{fig:mono}
\end{figure}

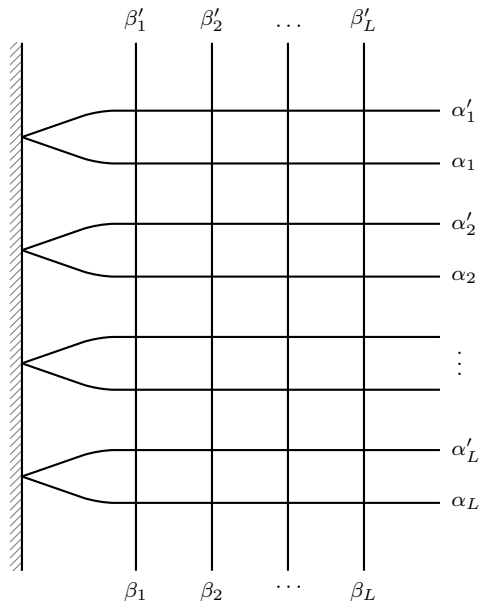
\begin{figure}
\begin{center}
\begin{tikzpicture}[font=\scriptsize,thick]
	\draw[rounded corners=6pt] (0,1.25) -- (1,1.25-.35) -- (5.5,1.25-.35) node[right]{$\alpha_L$};
	\draw[rounded corners=6pt] (0,1.25) -- (1,1.25+.35) -- (5.5,1.25+.35) node[right]{$\alpha_L'$};
	\draw[rounded corners=6pt] (0,2.75) -- (1,2.75-.35) -- (5.5,2.75-.35);
	\draw[rounded corners=6pt] (0,2.75) -- (1,2.75+.35) -- (5.5,2.75+.35);
	\draw (5.5,2.75) node[right,yshift=3pt]{$\:\vdots$};
	\draw[rounded corners=6pt] (0,4.25) -- (1,4.25-.35) -- (5.5,4.25-.35) node[right]{$\alpha_2$};
	\draw[rounded corners=6pt] (0,4.25) -- (1,4.25+.35) -- (5.5,4.25+.35) node[right]{$\alpha_2'$};
	\draw[rounded corners=6pt] (0,5.75) -- (1,5.75-.35) -- (5.5,5.75-.35) node[right]{$\alpha_1$};
	\draw[rounded corners=6pt] (0,5.75) -- (1,5.75+.35) -- (5.5,5.75+.35) node[right]{$\alpha_1'$};
	\draw (1+.5,0) node[below]{$\beta_1$} -- (1+.5,7) node[above]{$\beta_1'$};
	\draw (2+.5,0) node[below]{$\beta_2$} -- (2+.5,7) node[above]{$\beta_2'$};
	\draw (3+.5,0) node[below,xshift=1pt]{$\cdots$} -- (3+.5,7) node[above,xshift=1pt]{$\cdots$};
	\draw (4+.5,0) node[below]{$\beta_L$} -- (4+.5,7) node[above]{$\beta_L'$};
	\fill[preaction={fill,white},pattern=north east lines, pattern color=gray] (0,0) rectangle (-.15,7) ; \draw (0,0) -- (0,7);
\end{tikzpicture}
\end{center}
\caption{Representation of the partition function of the six-vertex model with one reflecting end and domain-wall boundaries.
In this work we have $\alpha_j , \beta_j = 1$ and $\alpha_j' , \beta_j' = 2$.}
\label{fig:partition}
\end{figure}

\section{Algebraic-functional approach}
\label{sec:AF}

In the works \cite{Galleas_2008, Galleas_2010, Galleas_2011, Galleas_2012, Galleas_2013, Galleas_SCP, Galleas_proc}
we have described a mechanism yielding functional equations satisfied by quantities of physical interest 
as a direct consequence of the Yang-Baxter algebra. This approach has been employed for the determination of
spectra \cite{Galleas_2008, Galleas_2013} and partition functions \cite{Galleas_2012, Galleas_2013}
of integrable vertex models. One issue arising within this method is that the algebraic relations we are considering
might not suffice to determine the desired quantities.
Furthermore, it would be desirable to have the simplest possible equations such that finding its solutions
can be achieved without much effort. Up to the present moment, we have only considered the Yang-Baxter algebra and its dynamical
counterpart as a source of functional relations \cite{Galleas_proc} and here we aim to show that
the reflection algebra (\ref{REA}) can also be exploited along the same lines. For this we need to introduce
the following definitions.

\begin{mydef} \label{pidef}
Let $\mathcal{M}(\lambda) \coloneqq \{ \mathcal{A}, \mathcal{B} , \mathcal{C} , \mathcal{D} \}(\lambda)$
and define $\mathcal{W}_n \coloneqq \mathcal{M}(\lambda_1) \times \mathcal{M}(\lambda_2) \times \dots \times \mathcal{M}(\lambda_n)$
with $n$-tuples $(\chi_1 , \dots , \chi_n)$ understood as $\mathop{\overrightarrow\prod}\limits_{1 \le j \le n } \chi_j$.
Also, let $\mathbb{C}[ \lambda_1^{\pm 1} , \lambda_2^{\pm 1} , \dots , \lambda_n^{\pm 1} ]$ be the ring of
meromorphic functions in the variables $\lambda_1, \dots , \lambda_n$ and define 
$\tilde{\mathcal{W}}_n \coloneqq \mathbb{C}[ \lambda_1^{\pm 1} , \dots , \lambda_n^{\pm 1} ] \otimes \gen{span}_{\mathbb{C}} ( \mathcal{W}_n )$.
\end{mydef}

To obtain functional relations from the reflection algebra we also need to introduce an appropriate linear map 
\[
\label{pin}
\gen{\pi}_n \colon  \tilde{\mathcal{W}}_n \to \mathbb{C}[ \lambda_1^{\pm 1} , \lambda_2^{\pm 1} , \dots , \lambda_n^{\pm 1} ] \; .
\]
A suitable realization of (\ref{pin}) will be given shortly.

\paragraph{Reflection relation of degree $n$.} The reflection algebra (\ref{REA}) encodes a set of sixteen commutation
relations governing the elements of (\ref{ABCD}). It is clear from (\ref{REA}) that those commutation
rules are quadratic and here they are referred to as reflection relations of degree two. The repeated use
of (\ref{REA}) then yields relations in $\tilde{\mathcal{W}}_n$ which shall be referred to as reflection relations
of degree $n$.

\subsection{Functional equations}
\label{sec:FZ}

In Definition \ref{pidef} we have introduced a map $\gen{\pi}_n$ assigning multivariate complex functions
to the elements of the set $\mathcal{W}_{n}$. Here our goal is to evaluate the partition function (\ref{PF}), and this
can be achieved from the study of suitable functional equations derived through the application of the map (\ref{pin})
on reflection relations of higher degree.
This procedure will require the following ingredients: a suitable realization of the map $\gen{\pi}_n$ and a convenient
reflection-algebra relation. As a matter of fact, different functional relations can be derived for the partition function
$\mathcal{Z}$ by changing these ingredients. 

\paragraph{Realization of $\gen{\pi}_n$.} The operatorial formulation of the partition function $\mathcal{Z}$ as 
given by (\ref{PF}) suggests that a suitable realization of $\gen{\pi}_n$ is given by the following scalar product:
\[
\label{pir}
\gen{\pi}_n (\mathcal{F}) \coloneqq \bra{\bar{0}} \mathcal{F} \ket{0} \; ,  
\]
for $\mathcal{F} \in \tilde{\mathcal{W}}_n$ and vectors $\ket{0}, \ket{\bar{0}} \in \mathbb{V}_{\mathcal{Q}}$ defined in (\ref{zero}).

\paragraph{Reflection-algebra relation.} Next we look for appropriate reflection relations of higher degree from which we can find functional relations
satisfied by the partition function $\mathcal{Z}$. In order to build such higher-degree relations we start
from the following fundamental commutation rules contained in (\ref{REA}):
\<
\label{ADB}
\mathcal{A}(\lambda_1) \mathcal{B}(\lambda_2) &=& \frac{a(\lambda_2 - \lambda_1)}{b(\lambda_2 - \lambda_1)} \frac{b(\lambda_2 + \lambda_1)}{a(\lambda_2 + \lambda_1)} \mathcal{B}(\lambda_2) \mathcal{A}(\lambda_1) - \frac{b(2 \lambda_2)}{a(2 \lambda_2)} \frac{c(\lambda_2 - \lambda_1)}{b(\lambda_2 - \lambda_1)} \mathcal{B}(\lambda_1) \mathcal{A}(\lambda_2) \nonumber \\
&&- \frac{c(\lambda_2 + \lambda_1)}{a(\lambda_2 + \lambda_1)} \mathcal{B}(\lambda_1) \tilde{\mathcal{D}}(\lambda_2) \nonumber \\
\tilde{\mathcal{D}}(\lambda_1) \mathcal{B}(\lambda_2) &=& \frac{a(\lambda_2 + \lambda_1 + \gamma)}{b(\lambda_2 + \lambda_1 + \gamma)} \frac{a(\lambda_1 - \lambda_2)}{b(\lambda_1 - \lambda_2)} \mathcal{B}(\lambda_2) \tilde{\mathcal{D}}(\lambda_1) - \frac{a(2 \lambda_1 + \gamma)}{b(2 \lambda_1 + \gamma)} \frac{c(\lambda_1 - \lambda_2)}{b(\lambda_1 - \lambda_2)} \mathcal{B}(\lambda_1) \tilde{\mathcal{D}}(\lambda_2) \nonumber \\
&&+ \frac{b(2 \lambda_2)}{a(2 \lambda_2)} \frac{a(2 \lambda_1 + \gamma)}{b(2 \lambda_1 + \gamma)} \frac{c(\lambda_2 + \lambda_1)}{a(\lambda_2 + \lambda_1)} \mathcal{B}(\lambda_1) \mathcal{A}(\lambda_2) \nonumber \\
\mathcal{B}(\lambda_1) \mathcal{B}(\lambda_2) &=& \mathcal{B}(\lambda_2) \mathcal{B}(\lambda_1) \; .
\>
Note that the above relation is given in terms of the operator $\tilde{\mathcal{D}}(\lambda) = \mathcal{D}(\lambda)  - \frac{c(2 \lambda)}{a(2 \lambda)} \mathcal{A}(\lambda)$.

Next we describe a suitable functional relation satisfied by (\ref{PF}). Although the partition function $\mathcal{Z}$ is a multivariate function
depending on $L$ spectral parameters $\lambda_i$, in addition to parameters $\mu_i$, $h$ and $\gamma$, here we shall obtain a functional equation
determining $\mathcal{Z}$ where only $\lambda_i$ play the role of variables. 

\begin{theorem} \label{theo_fun}
The partition function of the six-vertex model with one reflecting end and domain-wall boundaries obeys the
functional equation
\[
\label{TPA}
M_0 \; \mathcal{Z}(\lambda_1 , \dots , \lambda_L) + \sum_{i=1}^L M_i \; \mathcal{Z}(\lambda_0, \lambda_1 , \dots , \lambda_{i-1} , \lambda_{i+1} , \dots , \lambda_L ) = 0 \; ,
\]
with coefficients $M_0$ and $M_i$ given by
\<
\label{M0MI}
M_0 &\coloneqq & \bar{\Lambda}_{\mathcal{A}} (\lambda_0) - \Lambda_{\mathcal{A}} (\lambda_0) \prod_{j=1}^{L} \frac{a(\lambda_j - \lambda_0)}{b(\lambda_j - \lambda_0)} \frac{b(\lambda_j + \lambda_0)}{a(\lambda_j + \lambda_0)} \nonumber \\
M_i &\coloneqq & \frac{b(2 \lambda_i)}{a(2 \lambda_i)} \frac{c(\lambda_i - \lambda_0)}{b(\lambda_i - \lambda_0)} \Lambda_{\mathcal{A}} (\lambda_i) \prod_{\substack{j=1 \\j \neq i}}^{L} \frac{a(\lambda_j - \lambda_i)}{b(\lambda_j - \lambda_i)} \frac{b(\lambda_j + \lambda_i)}{a(\lambda_j + \lambda_i)} \nonumber \\
&& + \frac{c(\lambda_i + \lambda_0)}{a(\lambda_i + \lambda_0)} \Lambda_{\tilde{\mathcal{D}}} (\lambda_i) \prod_{\substack{j=1 \\j \neq i}}^{L} \frac{a(\lambda_i - \lambda_j)}{b(\lambda_i - \lambda_j)} \frac{a(\lambda_i + \lambda_j + \gamma)}{b(\lambda_i + \lambda_j + \gamma)} \; .
\>
The functions $\Lambda_{\mathcal{A}}$, $\Lambda_{\tilde{\mathcal{D}}}$ and $\bar{\Lambda}_{\mathcal{A}}$ were defined in (\ref{lambda}).
\end{theorem}
\begin{proof}
Consider the following element of $\mathcal{W}_{n+1}$,
\[
\label{ABBi}
\mathcal{A}(\lambda_0) \mathop{\overrightarrow\prod}\limits_{1 \le j \le n } \mathcal{B}(\lambda_j) \; ,
\]
under the light of the reflection algebra (\ref{REA}). The repeated use of (\ref{ADB}) yields the following
reflection relation of order $n+1$,
\<
\label{ABB}
\mathcal{A}(\lambda_0) \mathop{\overrightarrow\prod}\limits_{1 \le j \le n } \mathcal{B}(\lambda_j) &=& \prod_{j=1}^{n} \frac{a(\lambda_j - \lambda_0)}{b(\lambda_j - \lambda_0)} \frac{b(\lambda_j + \lambda_0)}{a(\lambda_j + \lambda_0)} \mathop{\overrightarrow\prod}\limits_{1 \le j \le n } \mathcal{B}(\lambda_j) \mathcal{A}(\lambda_0)  \nonumber \\
&& - \sum_{i=1}^{n} \frac{b(2 \lambda_i)}{a(2 \lambda_i)} \frac{c(\lambda_i - \lambda_0)}{b(\lambda_i - \lambda_0)}  \prod_{\substack{j=1 \\ j \neq i}}^{n} \frac{a(\lambda_j - \lambda_i)}{b(\lambda_j - \lambda_i)} \frac{b(\lambda_j + \lambda_i)}{a(\lambda_j + \lambda_i)} \mathop{\overrightarrow\prod}\limits_{\substack{0 \le j \le n \\ j \neq i}} \mathcal{B}(\lambda_j) \mathcal{A}(\lambda_i) \nonumber \\
&& - \sum_{i=1}^{n} \frac{c(\lambda_i + \lambda_0)}{a(\lambda_i + \lambda_0)} \prod_{\substack{j=1 \\ j \neq i}}^{n} \frac{a(\lambda_i - \lambda_j)}{b(\lambda_i - \lambda_j)} \frac{a(\lambda_i + \lambda_j + \gamma)}{b(\lambda_i + \lambda_j + \gamma)} \mathop{\overrightarrow\prod}\limits_{\substack{0 \le j \le n \\ j \neq i}} \mathcal{B}(\lambda_j) \tilde{\mathcal{D}}(\lambda_i) \; . \nonumber \\
\> 
Next we set $n=L$ and apply the map $\gen{\pi}_{L+1}$ given by (\ref{pir}) to (\ref{ABB}).
The left-hand side of (\ref{ABB}) then yields the term $\gen{\pi}_{L+1} (\mathcal{A}(\lambda_0) \mathop{\overrightarrow\prod}\limits_{1 \le j \le L } \mathcal{B}(\lambda_j))$ while the right-hand side
produces terms of the form $\gen{\pi}_{L+1} ( \mathop{\overrightarrow\prod}\limits_{1 \le j \le L } \mathcal{B}(\nu_j) \mathcal{A}(\nu) )$
and $\gen{\pi}_{L+1} ( \mathop{\overrightarrow\prod}\limits_{1 \le j \le L } \mathcal{B}(\nu_j) \tilde{\mathcal{D}}(\nu) )$.
Note that $\gen{\pi}_{L+1}$ reduces to $\gen{\pi}_{L}$ due to the $\alg{sl}(2)$ highest/lowest weight properties
exhibited by the realization (\ref{pir}). More precisely we have:
\<
\label{red}
\gen{\pi}_{L+1} (\mathcal{A}(\lambda_0) \mathop{\overrightarrow\prod}\limits_{1 \le j \le L } \mathcal{B}(\lambda_j)) &=& \bar{\Lambda}_{\mathcal{A}}(\lambda_0) \gen{\pi}_{L} (\mathop{\overrightarrow\prod}\limits_{1 \le j \le L } \mathcal{B}(\lambda_j)) \; , \nonumber  \\
\gen{\pi}_{L+1} ( \mathop{\overrightarrow\prod}\limits_{1 \le j \le L } \mathcal{B}(\nu_j) \mathcal{A}(\nu) ) &=& \Lambda_{\mathcal{A}} (\nu) \gen{\pi}_{L} ( \mathop{\overrightarrow\prod}\limits_{1 \le j \le L } \mathcal{B}(\nu_j) ) \; , \nonumber \\
\gen{\pi}_{L+1} ( \mathop{\overrightarrow\prod}\limits_{1 \le j \le L } \mathcal{B}(\nu_j) \tilde{\mathcal{D}}(\nu) ) &=& \Lambda_{\tilde{\mathcal{D}}} (\nu) \gen{\pi}_{L} ( \mathop{\overrightarrow\prod}\limits_{1 \le j \le L } \mathcal{B}(\nu_j) ) \; .
\>
Now we can identify the partition function $\mathcal{Z} (\lambda_1 , \dots , \lambda_L) = \gen{\pi}_{L} (\mathop{\overrightarrow\prod}\limits_{1 \le j \le L } \mathcal{B}(\lambda_j))$
on the right-hand side of (\ref{red}). Thus the relations (\ref{ABB}) and (\ref{red}) under the above mentioned conditions result in the functional equation (\ref{TPA}).
This proves Theorem \ref{theo_fun}.
\end{proof}

\begin{remark} \label{multi}
The functional equation (\ref{TPA}) is invariant under the permutation of variables $\lambda_i \leftrightarrow \lambda_j$ 
for $i,j \in \{1,2, \dots , L \}$. This conclusion follows directly from Lemma \ref{symm_lemma} which will be stated below.
However, the permutation $\lambda_0 \leftrightarrow \lambda_j$ yields a different functional equation for $\mathcal{Z}$. 
The resulting equation exhibits the same structure as (\ref{TPA}), with modified coefficients though. In this way (\ref{TPA})
actually encodes a set of $L+1$ equations.  
\end{remark}

\subsection{The partition function $\mathcal{Z}$}
\label{sec:sol}

This section is devoted to the determination of the partition function (\ref{PF}) 
as a particular solution of the functional equation (\ref{TPA}). A priori we do not have any
guarantee that (\ref{TPA}) is enough for that but direct inspection reveals that this is indeed
the case for small values of the lattice length $L$.

The general strategy for solving (\ref{TPA}) will follow the same steps described in \cite{Galleas_2013}.
This is anticipated since the structure of (\ref{TPA}) resembles that of the functional equation derived in
\cite{Galleas_2013} for the partition function of the elliptic SOS model with domain-wall boundaries.
However, here we shall need to exploit some further properties of (\ref{TPA}) which were not required in 
\cite{Galleas_2013}. In order to clarify our methodology let us first stress some characteristics of our functional equation. 
Firstly, Eq. (\ref{TPA}) is an equation for a complex multivariate function $\mathcal{Z}$ formed by a linear combination of terms containing 
$\mathcal{Z}(\lambda_1 , \lambda_2 , \dots , \lambda_L)$ with one of the variables $\lambda_i$ replaced by the variable $\lambda_0$. 
Thus (\ref{TPA}) runs over the set of variables $\{ \lambda_0 , \lambda_1 , \dots ,\lambda_L \}$. 
In addition to that, our equation is homogeneous in the sense that if $\mathcal{Z}$ is a solution then
$\omega \mathcal{Z}$ also solves (\ref{TPA}) for any $\omega \in \mathbb{C}$ independent of the variables $\lambda_i$.
This property anticipates that we shall need to evaluate the partition function (\ref{PF}) for a particular value of its variables 
in order to having the desired solution completely fixed. Moreover, due to the linearity of Eq. (\ref{TPA}), we need to address the question
of uniqueness of the solution. The partition function (\ref{PF}) consists of a particular polynomial solution and the uniqueness within such class
of solutions was proved in \cite{Galleas_2012} under very general conditions. 

Considering the above discussion the following lemmas will assist us through the determination of the partition function $\mathcal{Z}$.

\begin{lemma}[Polynomial structure] \label{pol_lemma}
The partition function $\mathcal{Z}$ defined in (\ref{PF}) is of the form $\mathcal{Z}(\lambda_1, \dots , \lambda_L) = \bar{\mathcal{Z}}(x_1 , \dots , x_L) \prod_{i=1}^{L} x_i^{-L}$, where 
$x_i \coloneqq e^{2 \lambda_i}$ and $\bar{\mathcal{Z}}(x_1 , \dots , x_L)$ is a polynomial of degree $2L$ in each of its variables.
\end{lemma}
\begin{proof}
The proof is obtained by induction and can be found in \Appref{sec:pol}.
\end{proof}

\begin{lemma} \label{symm_lemma}
Analytic solutions of (\ref{TPA}) are symmetric functions. More precisely, they satisfy the 
property $\mathcal{Z} (\dots, \lambda_i , \dots , \lambda_j , \dots) = \mathcal{Z} (\dots, \lambda_j , \dots , \lambda_i , \dots)$.
\end{lemma}
\begin{proof}
This property follows from the structure of poles appearing in (\ref{M0MI}). See \Appref{sec:symm} for details. 
\end{proof}

\begin{lemma}[Special zeroes] \label{zeroes_lemma}
For $L\geq 2$ the partition function $\mathcal{Z}$ vanishes for the specialization of variables $\lambda_1 = \mu_1 - \gamma$
and $\lambda_2 = \mu_1$. The same holds for the specialization $\lambda_1 = \mu_1 - \gamma$ and $\lambda_2 = -\mu_1 - \gamma$.
\end{lemma}
\begin{proof}
The proof follows from the inspection of (\ref{TPA}) under these specializations of variables, taking into account Remark \ref{multi}.
See \Appref{sec:zeroes} for details.
\end{proof}

\begin{lemma}[Asymptotic behavior] \label{asymp_lemma}
In the limit where all variables $x_i \rightarrow \infty$, the function $\bar{\mathcal{Z}}$ behaves as
\[
\label{asymp_form}
\bar{\mathcal{Z}} \sim \frac{q^{\frac{L(L-1)}{2}}}{2^{L(2L+1)}} (q - q^{-1})^L [L!]_{q^2} \prod_{i=1}^{L} (t y_i^{-\frac{1}{2}} - t^{-1} y_i^{\frac{1}{2}}) x_i^{2 L} \; ,
\]
where $q \coloneqq e^{\gamma}$, $t \coloneqq e^{h}$, $y_i \coloneqq e^{2 \mu_i}$ and $[n!]_{q^2} \coloneqq 1 (1 + q^2) (1 + q^2 + q^4) \dots (1 + q^2 + \dots + q^{2(n-1)})$ is the $q$-factorial function.
\end{lemma}
\begin{proof}
As $x_i \rightarrow \infty$ the generators (\ref{ABCD}) tend to the generators of the $\mathcal{U}_q [\alg{sl}(2)]$ algebra.
The properties of the latter can be employed to demonstrate (\ref{asymp_form}) as is shown in \Appref{sec:asymp}.
\end{proof}

\begin{remark} \label{add1}
Due to Lemma \ref{symm_lemma}, Eq. (\ref{TPA}) can also be written in a more compact form using the notation $\mathcal{Z}(\lambda_1 , \dots , \lambda_L) = \mathcal{Z}( X^{1,L} )$
and $\mathcal{Z}(\lambda_0, \lambda_1 , \dots , \lambda_{i-1} , \lambda_{i+1} , \dots , \lambda_L ) = \mathcal{Z}( X^{0,L}_i )$ where 
$X^{i,j}  \coloneqq  \{ \lambda_k \; : \; i \leq k \leq j \}$ and $X^{i,j}_l \coloneqq X^{i,j} \setminus \{ \lambda_l \}$. 
\end{remark}

\subsubsection{Multiple integral representation}
\label{sec:solA}

The resolution of (\ref{TPA}) will follow a sequence of systematic steps based on Lemmas \ref{pol_lemma} to \ref{asymp_lemma}.
The general procedure consists in finding suitable specializations of the variables $\lambda_0$ and $\lambda_L$ allowing us to 
invoke the above lemmas. The desired solution of (\ref{TPA}) is given by the following theorem.

\begin{theorem} \label{SOLa}
The partition function of the six-vertex model with one reflecting end and domain-wall boundaries (\ref{PF})
can be written as
\<
\label{theorA}
\mathcal{Z} (X^{1,L}) &=& c^L \oint \dots \oint \prod_{i=1}^{L} \frac{\dd w_i}{ 2 \pi \ii} \frac{\prod_{1 \leq i < j \leq L} a(\mu_i + w_j) b(\mu_i - w_j) b(w_i - w_j)^2}{\prod_{i,j=1}^{L} b(w_i - \lambda_j)} \nonumber \\
&& \qquad \qquad \quad \times \; \prod_{i=1}^{L} \frac{b(2 w_i)}{a(2 w_i)} \frac{b(h - \mu_i)}{b(h + \mu_i)} \Theta_i  \; , \nonumber
\>
where 
\<
\label{Theta_theor}
\Theta_i &\coloneqq &  \frac{b(w_i + h)}{a(w_i - \mu_i)} \prod_{j=i}^L a(w_i - \mu_j) a(w_i + \mu_j) \prod_{k=i+1}^L \frac{a(w_k - w_i)}{b(w_k - w_i)} \frac{b(w_k + w_i)}{a(w_k + w_i)}   \nonumber \\
&& - \; \frac{a(w_i - h)}{b(w_i + \mu_i)} \prod_{j=i}^L b(w_i - \mu_j) b(w_i + \mu_j) \prod_{k=i+1}^L \frac{a(w_i - w_k)}{b(w_i - w_k)} \frac{a(w_i + w_k + \gamma)}{b(w_i + w_k + \gamma)} \; . \nonumber \\
\>
\end{theorem}
\begin{proof} The proof of Theorem \ref{SOLa} follows from the resolution of (\ref{TPA}) 
taking into account certain properties of (\ref{PF}). The procedure consists of three steps.

\textit{Step $1$.} We first set $\lambda_0 = \mu_1 - \gamma$ in Eq. (\ref{TPA}). Under this specialization
the coefficient $M_0$ is reduced to a single product. This specialization also produces terms of the form 
$\mathcal{Z} ( \bar{X}^{2,L} )$ where $\bar{X}^{i,j} \coloneqq \{ \mu_1 - \gamma  \} \cup X^{i,j}$. 
Now due to Lemmas \ref{pol_lemma} to \ref{zeroes_lemma} we can write
\[
\mathcal{Z} ( \bar{X}^{2,L} ) = \prod_{j=2}^{L} b(\lambda_j - \mu_1) a(\lambda_j + \mu_1) \; \mathcal{V}(X^{2,L}) \; ,
\]
where $\mathcal{V}$ is a polynomial of degree $2(L-1)$ in each variable $x_i$ up to an overall exponential factor. 
Thus this particular specialization yields the expression
\[
\label{ZV}
\mathcal{Z} ( X^{1,L} ) = \kappa^{-1} \sum_{i=1}^L \frac{b(2 \lambda_i)}{a(2 \lambda_i)} \prod_{\substack{j=1 \\ j \neq i}}^{L} b(\lambda_j - \mu_1) a(\lambda_j + \mu_1) \; m_i \; \mathcal{V} (X^{1,L}_i) \; ,
\]
with coefficients $\kappa$ and $m_i$ given by
\<
\label{kami}
\kappa &\coloneqq & b(h + \mu_1) b(2 \mu_1 - 2 \gamma) \prod_{j=2}^{L} b(\mu_1 - \mu_j - \gamma) b(\mu_1 + \mu_j - \gamma) \nonumber \\
m_i &\coloneqq &  \frac{b(\lambda_i + h)}{a(\lambda_i - \mu_1)} \prod_{j=1}^L a(\lambda_i - \mu_j) a(\lambda_i + \mu_j) \prod_{\substack{k=1 \\ k \neq i}}^L \frac{a(\lambda_k - \lambda_i)}{b(\lambda_k - \lambda_i)} \frac{b(\lambda_k + \lambda_i)}{a(\lambda_k + \lambda_i)}   \nonumber \\
&& - \; \frac{a(\lambda_i - h)}{b(\lambda_i + \mu_1)} \prod_{j=1}^L b(\lambda_i - \mu_j) b(\lambda_i + \mu_j) \prod_{\substack{k=1 \\ k \neq i}}^L \frac{a(\lambda_i - \lambda_k)}{b(\lambda_i - \lambda_k)} \frac{a(\lambda_i + \lambda_k + \gamma)}{b(\lambda_i + \lambda_k + \gamma)} \; . \nonumber \\
\>

\textit{Step $2$.} We substitute formula (\ref{ZV}) back into the original equation (\ref{TPA}). By doing so
we are left with an equation involving only functions $\mathcal{V}$. Next we set $\lambda_L = \mu_1$ in the resulting
equation which then further simplifies to
\[
\label{TPAr}
\tilde{M}_0 \mathcal{V}(X^{1,L-1}) + \sum_{i=1}^{L-1} \tilde{M}_i \; \mathcal{V}(X^{0,L-1}_i) = 0 \; .
\]
The explicit form of the coefficients $\tilde{M}_0$ and $\tilde{M}_i$ is not enlightening but it is worth remarking
that for $L=2$ we find that (\ref{TPAr}) corresponds to (\ref{TPA}) with $L=1$ and $\mu_1$ replaced by $\mu_2$. 
This fact suggests that (\ref{TPAr}) should coincide with (\ref{TPA}) after replacing $L$ by $L-1$ and
$\mu_i$ by $\mu_{i+1}$. Unfortunately this is not the case for general values of $L$ and we actually find that
(\ref{TPAr}) consists of a linear combination of (\ref{TPA}) along the lines of Remark \ref{multi}.
Nevertheless, this still ensures that $\mathcal{V}$ is essentially our partition function
under the maps $L \mapsto L-1$ and $\mu_i \mapsto \mu_{i+1}$ since polynomial solutions are unique.

\textit{Step $3$.} The results of Step $2$ allows us to obtain an explicit representation for our 
partition function from the relation (\ref{ZV}) in a recursive manner. In fact, formula (\ref{ZV}) suggests
the following ansatz for $\mathcal{Z}$
\[
\label{int}
\mathcal{Z} (X^{1,L}) = \oint \dots \oint \prod_{i=1}^{L} \frac{\dd w_i}{ 2 \pi \ii} \frac{H(w_1, \dots , w_L)}{\prod_{i,j=1}^{L} b(w_i - \lambda_j)} \; ,
\] 
where $H$ is a function yet to be determined. In particular, here we also assume that the integration contours in (\ref{int})
enclose all the poles at $w_i = \lambda_j$ and that $H$ contains no poles inside those integration contours.
Then we consider the mechanism described in \cite{Galleas_SCP} to find the following relation determining the function $H$,
\<
\label{HH}
&& H(w_1 , \dots , w_L) = \nonumber \\
&& \frac{\bar{H}(w_2 , \dots , w_L)}{b(h + \mu_1)} \frac{b(2 w_1)}{a(2 w_1)} \prod_{j=2}^{L} b(w_1 - w_j)^2 b(\mu_1 - w_j) a(\mu_1 + w_j)  \nonumber \\
&& \times \left[ b(2 \mu_1 - 2 \gamma) \prod_{j=2}^{L} b(\mu_1 - \mu_j - \gamma) b(\mu_1 + \mu_j - \gamma) \right]^{-1} \nonumber \\
&& \times \left\{  \frac{b(w_1 + h)}{a(w_1 - \mu_1)} \prod_{j=1}^L a(w_1 - \mu_j) a(w_1 + \mu_j) \prod_{k=2}^L \frac{a(w_k - w_1)}{b(w_k - w_1)} \frac{b(w_k + w_1)}{a(w_k + w_1)} \right.  \nonumber \\
&& \quad \quad - \left. \frac{a(w_1 - h)}{b(w_1 + \mu_1)} \prod_{j=1}^L b(w_1 - \mu_j) b(w_1 + \mu_j) \prod_{k=2}^L \frac{a(w_1 - w_k)}{b(w_1 - w_k)} \frac{a(w_1 + w_k + \gamma)}{b(w_1 + w_k + \gamma)} \right\} \; . \nonumber \\
\>
The function $\bar{H}$ in (\ref{HH}) corresponds to $H$ under the maps $L \mapsto L-1$, $\mu_i \mapsto \mu_{i+1}$
up to an overall constant factor. In this way the relation (\ref{HH}) can be iterated once we know the
function $H(w_1)$. This function can be directly read from the solution of (\ref{TPA}) for $L=1$ which can be found
in \Appref{sec:L1}. Thus the iteration of (\ref{HH}) yields the following expression for the function $H$,
\<
\label{Hgen}
H(w_1 , \dots , w_L) = c^L \; \prod_{i=1}^{L} \frac{b(2 w_i)}{a(2 w_i)} \frac{b(h - \mu_i)}{b(h + \mu_i)} \Theta_i \; \prod_{1 \leq i < j \leq L} a(\mu_i + w_j) b(\mu_i - w_j) b(w_i - w_j)^2 \; , \nonumber \\
\>
where $\Theta_i$ is given by (\ref{Theta_theor}). Formula (\ref{Hgen}) already takes into account the asymptotic behavior
stated in Lemma \ref{asymp_lemma} and this completes the proof of Theorem \ref{SOLa}.
\end{proof}

\subsection{Partial differential equations}
\label{sec:PDE}

In \Secref{sec:FZ} we have derived a functional equation governing the partition function (\ref{PF})
as a direct consequence of the reflection algebra (\ref{REA}) and the highest/lowest weight property of the 
vectors $\ket{0}$ and $\ket{\bar{0}}$. Some properties of our functional equation have already been discussed in
\Secref{sec:sol} and here we intend to demonstrate some further properties.
More precisely, in this section we shall unveil a set of linear partial differential equations underlying (\ref{TPA}). 
This type of hidden structure was first presented in \cite{Galleas_2011} for a similar type of equation and subsequently developed
in \cite{Galleas_proc, Galleas_2014}. The first step towards that description is to recast (\ref{TPA}) in an operatorial form.
This can be achieved with the help of the operator $D_{i}^{\alpha}$ defined as follows.

\begin{mydef} \label{dia_def}
Let $n \in \mathbb{Z}_{> 0}$ and $\alpha \in \mathbb{Z} \backslash \{1,2, \dots, n \}$. As before we write $\mathbb{C}[z_1^{\pm 1}, \dots , z_n^{\pm 1}]$ for
the space of meromorphic functions on $\mathbb{C}^n$. Now consider the following operator $D_{i}^{\alpha} \colon \mathbb{C}[z_1^{\pm 1}, \dots , z_i^{\pm 1} , \dots , z_n^{\pm 1}] \to \mathbb{C}[z_1^{\pm 1}, \dots , z_{\alpha}^{\pm 1} , \dots , z_n^{\pm 1}]$
defined by
\[
\label{Dia}
( D_{i}^{\alpha} f )(z_1, \dots , z_i , \dots , z_n) \coloneqq f(z_1, \dots , z_{\alpha} , \dots , z_n) \; .
\]
\end{mydef}
Definition \ref{dia_def} is clearly motivated by the structure of (\ref{TPA}) and it allows one 
to rewrite Eq. (\ref{TPA}) as $\mathfrak{L}(\lambda_0) \mathcal{Z}(X^{1,L}) = 0$ where
\[
\label{Lop}
\mathfrak{L}(\lambda_0) \coloneqq M_0 + \sum_{i=1}^L M_i \; D_i^{0} \; .
\]
In (\ref{Lop}) we have made the dependence of $\mathfrak{L}$ on $\lambda_0$ explicit to stress
that this reformulation concentrates the whole dependence of our functional equation on $\lambda_0$
in the operator $\mathfrak{L}$. In particular, this property will allow us to extract a set of
partial differential equations from (\ref{TPA}) due to the fact that there exists
a differential realization of (\ref{Dia}) when we restrict the action of the operator $D_i^{\alpha}$
to a particular function space. In order to describe this differential realization we first need
to introduce some extra definitions and conventions.

\begin{mydef} \label{fun_space}
Let $\mathbb{K}[z_1 , \dots , z_n]$ denote the multivariate polynomial ring in the variables $z_1, \dots , z_n$ with coefficients in an
arbitrary field $\mathbb{K}$. We will also use the abbreviation $\mathbb{K}[z] \coloneqq \mathbb{K}[z_1 , \dots , z_n]$. Using this shorthand notation, we define 
$\mathbb{K}_m [z] \subseteq \mathbb{K}[z]$ to be the subspace of $\mathbb{K}[z]$ formed by polynomials of degree $m$ in each variable $z_i$. 
\end{mydef}

\begin{lemma} \label{diff} 
The differential operator
\[
\label{Dia_diff}
D_i^{\alpha} = \sum_{k=0}^m \frac{(z_{\alpha} - z_i)^k}{k!} \frac{\partial^k }{\partial z_i^k}
\]
is a realization of (\ref{Dia}) on the space $\mathbb{K}_m [z]$.
\end{lemma}
\begin{proof}
The proof follows from the series expansion of functions in $\mathbb{K}_m [z]$. The details of this
analysis can be found in \cite{Galleas_2011, Galleas_proc}.
\end{proof}

The realization (\ref{Dia_diff}) can not be directly substituted in (\ref{Lop}) since the function
$\mathcal{Z}$ we are interested in does not belong to $\mathbb{K}_m [z]$. However, as far as the 
function $\bar{\mathcal{Z}}$ defined in Lemma \ref{pol_lemma} is concerned, we have that 
$\bar{\mathcal{Z}}(x_1 , \dots , x_L) \in \mathbb{K}_{2L} [x_1 , \dots , x_L]$ with
$\mathbb{K} = \mathbb{C}[y_1^{\pm 1}, \dots , y_L^{\pm 1}, q^{\pm 1} , t^{\pm 1}]$
and thus (\ref{Dia_diff}) can be employed. Here we use the notation of Lemma \ref{asymp_lemma}
where, in particular,  $x_j = e^{2 \lambda_j}$. We then define the rescaled coefficients 
\[
\label{Mbar}
\bar{M}_0 \coloneqq M_0 \prod_{j=1}^{L} x_j^{-L} \qquad \mbox{and} \qquad \bar{M}_i \coloneqq M_i \prod_{\substack{j=0 \\ j \neq i}}^{L} x_j^{-L} \; .
\]
In this way Eq. (\ref{TPA}) reads $\bar{\mathfrak{L}}(x_0) \bar{\mathcal{Z}}(X^{1,L}) = 0$ where
\[
\label{bLop}
\bar{\mathfrak{L}}(x_0) \coloneqq \bar{M}_0 + \sum_{i=1}^L \bar{M}_i \; D_i^{0} \; ,
\]
and $X^{i,j} = \{ x_k \; : \; i \leq k \leq j \}$ as in Remark \ref{add1}.

Now we can substitute (\ref{Dia_diff}) in (\ref{bLop}), and the next step of our analysis
is to look at the analytical properties of $\bar{\mathfrak{L}}(x_0)$ as function of $x_0$
or equivalently $\lambda_0$. The explicit expressions for the coefficients $\bar{M}_0$ and $\bar{M}_i$
are obtained from (\ref{M0MI}) and we can readily see that $\bar{\mathfrak{L}}$ contains simple poles
at the zeroes of $a(2 \lambda_0)$, $b(\lambda_0 - \lambda_i)$ and $a(\lambda_0 + \lambda_i)$. 
The residues of $\bar{\mathfrak{L}}$ at the poles $a(2 \lambda_0)=0$ and $b(\lambda_0 - \lambda_i) = 0$ vanish but
the same is not true for the poles at $a(\lambda_0 + \lambda_i)=0$. Thus (\ref{bLop}) is of the form
\[
\label{LopR}
\bar{\mathfrak{L}}(x_0) = \frac{x_0^{-\frac{(L+1)}{2}}}{\prod_{j=1}^L a(\lambda_0 + \lambda_j)} \bar{\mathfrak{L}}_R (x_0) \; ,
\]  
where $\bar{\mathfrak{L}}_R (x_0)$ has no poles for $x_0 \in \mathbb{C} \backslash \{0 \}$. Moreover, the direct inspection
of $\bar{\mathfrak{L}}_R$ reveals that it is indeed a polynomial of the form
\[
\label{LopR_pol}
\bar{\mathfrak{L}}_R (x_0) = \sum_{k=0}^{2L} x_0^k \; \gen{\Omega}_k  \; ,
\]
with differential-operator valued coefficients $\gen{\Omega}_k$. Now since $\bar{\mathfrak{L}}_R (x_0)$
is a polynomial, the equation $\bar{\mathfrak{L}}_R (x_0) \bar{\mathcal{Z}}(X^{1,L}) = 0$ must be satisfied by each power
of $x_0$ separately. In this way we are left with a total of $2L+1$ partial differential equations formally reading
\[
\label{omg_k}
\gen{\Omega}_k \; \bar{\mathcal{Z}}(X^{1,L}) = 0 \qquad \qquad 0 \leq k \leq 2L \; .
\]

\paragraph{The operator $\gen{\Omega}_{2 L}$.} Due to (\ref{Dia_diff}) and the fact that 
$\bar{\mathcal{Z}}(X^{1,L}) \in \mathbb{K}_{2L} [x_1 , \dots , x_L]$, the
differential operators $\gen{\Omega}_k$ are linear and contain partial derivatives with respect
to the variables $x_i$ of order ranging from $1$ to $2L$. Although the explicit form of the operators
$\gen{\Omega}_k$ for a given value of $L$ can be computed from (\ref{bLop}), (\ref{LopR}) and (\ref{LopR_pol}),
they mostly lead to cumbersome expressions which are not very enlightening. Fortunately, the situation for
the leading term operator $\gen{\Omega}_{2 L}$ is more interesting and we find the following compact expression,
\[
\label{O2L}
\gen{\Omega}_{2L} = \mathcal{U} + \sum_{i=1}^L \mathcal{Y}_i \; \frac{\partial^{2L}}{\partial x_i^{2L}} \; .
\]
The functions $\mathcal{U}$ and $\mathcal{Y}_i$ in (\ref{O2L}) explicitly read,
\<
\label{UY}
\mathcal{U} &\coloneqq & t^{-1} (1 - q^{2L}) + t \sum_{i=1}^L \left[ x_i q^2 + x_i^{-1} - (y_i + y_i^{-1})  \right] \nonumber \\
\mathcal{Y}_i &\coloneqq & - \frac{1}{(2L)!} \frac{\bar{a}_1 (x_i , x_i)}{\bar{a}_q (x_i , x_i)} \nonumber \\
&& \times \left\{ q \bar{a}_t (x_i , 1) \prod_{j=1}^L \bar{a}_q (x_i , y_j^{-1}) \bar{a}_q (x_i , y_j) \prod_{\substack{j=1 \\ j \neq i}}^{L} \frac{\bar{a}_q (x_j , x_i^{-1})}{\bar{a}_1 (x_j , x_i^{-1})} \frac{\bar{a}_1 (x_j , x_i)}{\bar{a}_q (x_j , x_i)}  \right. \nonumber \\
&& \left. \qquad + \;  \bar{a}_{q/t} (1, x_i) \prod_{j=1}^L \bar{a}_1 (x_i , y_j^{-1}) \bar{a}_1 (x_i , y_j) \prod_{\substack{j=1 \\ j \neq i}}^{L} \frac{\bar{a}_q (x_i , x_j^{-1})}{\bar{a}_1 (x_i , x_j^{-1})} \frac{\bar{a}_{q^2} (x_i , x_j)}{\bar{a}_q (x_i , x_j)} \right\} \; ,  \nonumber \\
\>
where $\bar{a}_{\omega} (x,y) \coloneqq x \omega - y^{-1} \omega^{-1}$.

Some comments are appropriate at this stage. To start with, the direct inspection of (\ref{O2L}) for small values of the 
lattice length $L$ reveals that our partial differential equation is fully able to determine the desired polynomial
solution up to an overall constant factor that is fixed by Lemma \ref{asymp_lemma}.
Moreover, the structure of (\ref{O2L}) resembles that of a quantum many-body hamiltonian with higher derivatives and we can regard the partition function $\bar{\mathcal{Z}}$ as the null-eigenvalue 
wave-function associated to $\gen{\Omega}_{2L}$. It is worth remarking here that a similar structure appeared previously 
for the standard six-vertex model with domain-wall boundaries in \cite{Galleas_proc}. In particular, the structure
of $\gen{\Omega}_{2L}$ is also shared by higher conserved quantities of the six-vertex model as demonstrated in \cite{Galleas_2014}.
To conclude we remark that although (\ref{O2L}) results in a differential equation of order $2L$, it can still be recasted
as a system of first-order equations using the reduction of order procedure. This analysis is explicitly performed
in \Appref{sec:RED}.

\section{Concluding remarks}
\label{sec:conclusion}

This work is mainly concerned with the interplay between functional equations and the reflection
algebra in the framework developed in \cite{Galleas_2010, Galleas_SCP, Galleas_proc}. More precisely, here
we have investigated the partition function of the six-vertex model with one reflecting end and domain-wall 
boundaries through this algebraic-functional approach. This methodology has been previously
considered for the  dynamical counterpart of the Yang-Baxter algebra in \cite{Galleas_2011, Galleas_2012},
and here we demonstrate the feasibility of the reflection algebra for that approach.
From this analysis we obtain functional relations satisfied by the partition function of the six-vertex
model with both domain-wall and reflecting boundaries.
Interestingly, the equation presented here exhibits the same structure as the one obtained in \cite{Galleas_2012, Galleas_proc}
for a partition function with simpler boundary conditions. Although \cite{Galleas_2012, Galleas_proc} and our present work
consider domain-wall boundary conditions, here we have also included a reflecting end, which makes this algebraic-functional analysis 
significantly more involved. However, the difference between the functional equations in those works and the present one is restricted
to the explicit form of their coefficients.

The starting point for the derivation of (\ref{TPA}) is the element (\ref{ABBi}) and the corresponding
reflection relation of higher degree (\ref{ABB}). This choice is arbitrary and we would have obtained a different equation
if we had started with a different element of $\mathcal{W}_n$.
For instance, the element $\tilde{\mathcal{D}}(\lambda_0) \mathop{\overrightarrow\prod}\limits_{1 \le j \le n } \mathcal{B}(\lambda_j)$
would have resulted in an equally simple functional equation. Here we have restricted our attention to the analysis of (\ref{TPA}) since this equation 
is already enough to determine the partition function.  

The solution of our equation is presented in \Secref{sec:sol} and is given in terms of a multiple-contour integral over
$L$ auxiliary variables. In contrast to the determinant representation obtained in \cite{Tsuchiya_1998}, our integral
formula offers the possibility of studying the homogeneous limit $\lambda_i \rightarrow \lambda$ and $\mu_i \rightarrow \mu$ 
straightforwardly. This feature seems to be of relevance for the analysis of the surface free energy of the $XXZ$ model
as discussed in \cite{Kozlowski_2012}. It is also important to remark here that the multiple integral formula given in
Theorem \ref{SOLa} can also be shown to satisfy the recurrence relations derived in \cite{Tsuchiya_1998}. Those recurrence
relations, in addition to extra properties, are able to uniquely characterize the model partition function
and thus can also be used to prove Theorem \ref{SOLa}. However, finding an explicit representation would still demand a very non-trivial
guess which is not required in our framework. In this sense the approach described here also offers a systematic
way of building explicit representations.

The structure of our functional equation is further studied in \Secref{sec:PDE} and we find interesting
properties which are not apparent at first sight. For instance, we shown that our equation actually
encodes a set of linear partial differential equations. Any single equation from this set is already able to determine
the model's partition function, and thus this set is simultaneously integrated. It is worth remarking here that this
property is a common feature exhibited by integrable hierarchies of differential equations.
In this work we have not analyzed the integrability of our partial differential equations in the classical sense,
but that direction certainly deserves further investigation. Our construction yields a total of $2L+1$ equations, involving among others
the differential operator (\ref{O2L}), whose structure resembles that of a quantum many-body hamiltonian with higher-order derivatives. 
Although the order of the corresponding differential equation depends on $L$, this equation can still be reformulated as a system of first-order
equations due to its linearity.

To conclude we remark here that partition functions with domain-wall boundaries and reflecting ends can also be formulated
for Solid-on-Solid models as described in \cite{Filali_2010, Filali_2011}. In that case the governing algebra
is a dynamical version of the reflection algebra and it would be interesting to investigate if our
approach can be extended to those cases.

\section{Acknowledgements}
\label{sec:ack}
This work is supported by the Netherlands Organization for Scientific Research (NWO) under the VICI grant 680-47-602 
and by the ERC Advanced grant research programme No. 246974, {\it ``Supersymmetry: a window to non-perturbative physics"}. 
The authors also thank the D-ITP consortium, a program of the Netherlands Organization for Scientific Research (NWO) funded by 
the Dutch Ministry of Education, Culture and Science (OCW).

\appendix

\section{Properties of $\ket{0}$ and $\ket{\bar{0}}$}
\label{sec:ZERO}

This appendix is devoted to the derivation of formulae (\ref{lambda}) arising as the eigenvalues
of the operators $\mathcal{A}(\lambda)$ and $\tilde{\mathcal{D}}(\lambda)$ with respect to the
vectors $\ket{0}$ and $\ket{\bar{0}}$ defined in (\ref{zero}). For that we shall make use of 
\eqref{full_mono} keeping in mind the representations \eqref{abcd} and \eqref{kmat}. In this way
the entries of (\ref{ABCD}) can be expressed as,
\<
\label{ABCD_via_abcd}
\mathcal{A}(\lambda)  &=& \kappa_{+}(\lambda)  A(\lambda) \bar{A}(\lambda)  + \kappa_{-}(\lambda)  B(\lambda) \bar{C}(\lambda) \nonumber \\
\mathcal{B}(\lambda)  &=& \kappa_{+}(\lambda)  A(\lambda) \bar{B}(\lambda)  + \kappa_{-}(\lambda)  B(\lambda) \bar{D}(\lambda)  \nonumber \\
\mathcal{C}(\lambda)  &=& \kappa_{+}(\lambda)  C(\lambda) \bar{A}(\lambda)  + \kappa_{-}(\lambda) D(\lambda) \bar{C}(\lambda)  \nonumber \\
\mathcal{D}(\lambda)  &=& \kappa_{+}(\lambda)  C(\lambda) \bar{B}(\lambda)  + \kappa_{-}(\lambda) D(\lambda) \bar{D}(\lambda)  \; ,
\>
recalling that $\kappa_{\pm}(\lambda) = b(h \pm \lambda)$. Here we are only interested in the operators $\mathcal{A}(\lambda)$
and $\mathcal{D}(\lambda)$, and from (\ref{ABCD_via_abcd}) we can see that (\ref{lambda}) can be computed from the action of
(\ref{abcd}) on the vectors (\ref{zero}). Due to the structure of (\ref{rmat}) and (\ref{mono}) we readily find the following relations
\begin{align} \label{acd_action}
A(\lambda)\ket{0} & = \prod_{j=1}^L a(\lambda-\mu_j) \ket{0} & \bar{A}(\lambda)\ket{0} &= \prod_{j=1}^L a(\lambda+\mu_j)\ket{0} \nonumber \\ 
D(\lambda)\ket{0} & = \prod_{j=1}^L b(\lambda-\mu_j) \ket{0} & \bar{D}(\lambda)\ket{0} &= \prod_{j=1}^L b(\lambda+\mu_j)\ket{0} \nonumber \\
C(\lambda)\ket{0} &= 0 & \bar{C}(\lambda) \ket{0} &= 0 \; ,
\end{align}
while an analogous computation yields
\begin{align}
\label{acd_action2} 
\bra{\bar{0}} A(\lambda) &= \prod_{j=1}^L b(\lambda-\mu_j)\bra{\bar{0}} & \bra{\bar{0}}\bar{A}(\lambda) &= \prod_{j=1}^L b(\lambda+\mu_j)\bra{\bar{0}} \nonumber \\
\bra{\bar{0}} D(\lambda) &= \prod_{j=1}^L a(\lambda-\mu_j)\bra{\bar{0}}  & \bra{\bar{0}} \bar{D}(\lambda) & = \prod_{j=1}^L a(\lambda+\mu_j)\bra{\bar{0}} \nonumber \\
\bra{\bar{0}} C(\lambda) &= 0  & \bra{\bar{0}}\bar{C}(\lambda) &= 0  \; .
\end{align}
In their turn, the action of $B(\lambda)$ and $\bar{B}(\lambda)$ on the vectors $\ket{0}$ and $\bra{\bar{0}}$ does not vanish
but they do not correspond to eigenvectors either.

Now turning our attention to the functions $\Lambda_{\mathcal{A}}$, $\Lambda_{\tilde{\mathcal{D}}}$ and $\bar{\Lambda}_{\mathcal{A}}$
described in \Secref{sec:CONV}, we can see that $\Lambda_{\mathcal{A}}$ can be directly read off from (\ref{ABCD_via_abcd}) and (\ref{acd_action}).
On the other hand, the evaluation of $\Lambda_{\tilde{\mathcal{D}}}$ is more involved as it corresponds to the eigenvalue of
the operator $\tilde{\mathcal{D}}(\lambda) = \mathcal{D}(\lambda)  - \frac{c(2 \lambda)}{a(2 \lambda)} \mathcal{A}(\lambda)$
with respect to the vector $\ket{0}$. The latter would then require the evaluation of $C(\lambda) \bar{B}(\lambda) \ket{0}$ as we can see from
(\ref{ABCD_via_abcd}). Fortunately the Yang-Baxter algebra (\ref{yba}) can help us with that computation. Due to the unitarity property
$\mathcal{R}(\lambda) \mathcal{R}(- \lambda) = a(\lambda) a(-\lambda) \mathbbm{1}$ we find the following algebraic relation
\[
\label{yba_ttbar}
\bar{\tau}_2 (\lambda) \, \mathcal{R}_{12} (2 \lambda) \, \tau_1 (\lambda) =  \tau_1 (\lambda) \, \mathcal{R}_{12} (2\lambda) \, \bar{\tau}_2 (\lambda) \; ,
\]
obtained from (\ref{yba}) under the specializations $\lambda_1 = - \lambda_2 = \lambda$. In particular, among the relations encoded in (\ref{yba_ttbar})
we have 
\[
\label{cbb}
C(\lambda) \bar{B}(\lambda) = \bar{B}(\lambda) C(\lambda) + \frac{c(2 \lambda)}{a(2 \lambda)} [ \bar{A}(\lambda) A(\lambda) - D (\lambda) \bar{D}(\lambda) ] \; ,
\]
which allows the evaluation of $C(\lambda) \bar{B}(\lambda) \ket{0}$ using (\ref{acd_action}).

To conclude we turn our attention to the computation of $\bar{\Lambda}_{\mathcal{A}}$ and from (\ref{ABCD_via_abcd})
we can see this would require the evaluation of  $\bra{\bar{0}} B(\lambda) \bar{C}(\lambda)$. The relations contained in 
(\ref{yba_ttbar}) are also helpful in that case. In particular we have the commutation relation
\[
\label{bcc}
B(\lambda) \bar{C}(\lambda) = \bar{C}(\lambda) B(\lambda) + \frac{c(2\lambda)}{a(2\lambda)} [ \bar{D}(\lambda) D(\lambda) - A(\lambda) \bar{A}(\lambda)] \; ,
\]
which yields the desired quantity with the help of (\ref{acd_action2}).

\section{Polynomial structure}
\label{sec:pol}

In this appendix we prove that the partition function defined in \eqref{PF} has the form stated in Lemma \ref{pol_lemma}.
More precisely, here we show that $\mathcal{Z}(\lambda_1, \dots , \lambda_L) = \bar{\mathcal{Z}}(x_1 , \dots , x_L) \prod_{i=1}^{L} x_i^{-L}$ where $\bar{\mathcal{Z}}(x_1 , \dots , x_L)$ 
is a polynomial of degree $2L$ in each one of the variables $x_i = e^{2 \lambda_i}$. For that it suffices to show that $\mathcal{B}$ has the form
\[
\label{B_L}
\mathcal{B}(x)=x^{-L}f_{\mathcal{B}}^{(2L)}(x) \; ,
\]
where $f_{\mathcal{B}}^{(2L)}(x) \in \mathbb{K}_{2L}[x] \otimes \gen{End}(\mathbb{V}_{\mathcal{Q}})$ with $\mathbb{K} = \mathbb{C}[y_1^{\pm 1}, \dots , y_L^{\pm 1}, q^{\pm 1} , t^{\pm 1}]$
in the notation of Definition \ref{fun_space}. In other words, $f_{\mathcal{B}}^{(2L)}(x)$ is a polynomial of degree $2L$
in the variable $x$, whose coefficients are products of meromorphic functions of $y_1, \dots , y_L, q, t$ and operators on
$\mathbb{V}_{\mathcal{Q}}$. Throughout this appendix we keep track of the degree of the polynomials by indicating it in superscript
as in (\ref{B_L}).  

The expression for $\mathcal{B}$ given in \eqref{ABCD_via_abcd} reduces our task to the analysis of the dependence of 
$\kappa_{\pm}$, $A$, $\bar{B}$, $B$ and $\bar{D}$ with $x$. From (\ref{kmat}) it is clear that $\kappa_{\pm}(x) = \pm \tfrac{1}{2} \, x^{-\frac{1}{2}}(x\, t^{\pm1} - t^{\mp1})$
and, therefore, it is enough to demonstrate that for a given $L$ we have
\begin{align}
\label{ab_L}
& A_L(x) = x^{-\frac{L}{2}}f_{A_L}^{(L)}(x)  && B_L(x) = x^{-\frac{L-1}{2}}f_{B_L}^{(L-1)}(x)  \\
\label{ab_L_bar}
& \bar{B}_L(x) = x^{-\frac{L-1}{2}}f_{\bar{B}_L}^{(L-1)}(x)  && \bar{D}_L(x) = x^{-\frac{L}{2}}f_{\bar{D}_L}^{(L)}(x) \; .
\end{align}
Here we consider $f^{(m)}_{\beta_L}(x) \in \mathbb{K}_m [x] \otimes \gen{End}(\mathbb{V}_{\mathcal{Q}})$ 
with $\mathbb{K} = \mathbb{C}[y_1^{\pm 1}, \dots , y_L^{\pm 1}, q^{\pm 1}]$ for $m \in \mathbb{N}$ and $\beta_L \in \{ A_L, B_L, C_L, D_L, \bar{A}_L , \bar{B}_L , \bar{C}_L , \bar{D}_L \}$. Also,
we have added the subscript $L$ to the elements of (\ref{abcd}) in order to emphasize the chain length
we are considering. Now we proceed to showing \eqref{ab_L} by induction on $L$. The expressions (\ref{ab_L_bar}) can be treated analogously.

For $L=1$ we notice that the matrices
\[
\label{KX}
K = \left( \begin{matrix} q & 0 \\ 0 & q^{-1} \end{matrix} \right) \; , \qquad X^{-} = \left( \begin{matrix} 0 & 0 \\ 1 & 0 \end{matrix} \right) 
\quad \mbox{and} \qquad X^{+} = \left( \begin{matrix} 0 & 1 \\ 0 & 0 \end{matrix} \right) 
\]
provide a two-dimensional representation of the $\mathcal{U}_q [\alg{sl}(2)]$ algebra obeying the commutation rules
\[
\label{Uqsl2}
K X^{\pm} K^{-1} = q^{\pm2} X^{\pm} \; ,  \qquad  [ X^{+},X^{-}] = \frac{K-K^{-1}}{q-q^{-1}} \; .
\]
Moreover, for $L=1$ the monodromy matrices (\ref{mono}) consist of a single $\mathcal{R}$-matrix. Thus by writing
(\ref{rmat}) in the auxiliary space as
\[
\label{rmat_L=1}
\mathcal{R} (\lambda-\mu_j) = \left( \begin{matrix}
A_1(\lambda) & B_1(\lambda) \\
C_1(\lambda) & D_1(\lambda) \end{matrix} \right) \; ,
\]
we have
\begin{align}
\label{abcd_L=1}
A_1(x) & = x^{-\frac{1}{2}} \, f_{A_1}^{(1)}(x) \coloneqq \tfrac{1}{2} \, x^{-\frac{1}{2}} \left(x\,q^{\frac{1}{2}} y_j^{-\frac{1}{2}} \, K^{\frac{1}{2}} - q^{-\frac{1}{2}} y_j^{\frac{1}{2}} \, K^{-\frac{1}{2}}\right)  \nonumber \\
B_1(x) & = f_{B_1}^{(0)}(x) \coloneqq \tfrac{1}{2}(q-q^{-1}) \, X^{-}  \nonumber \\
C_1(x) & = f_{C_1}^{(0)}(x) \coloneqq \tfrac{1}{2}(q-q^{-1}) \, X^{+}  \nonumber \\ 
D_1(x) & = x^{-\frac{1}{2}} \, f_{D_1}^{(1)}(x) \coloneqq \tfrac{1}{2} \, x^{-\frac{1}{2}}  \left(x\,q^{\frac{1}{2}} y_j^{-\frac{1}{2}} \, K^{-\frac{1}{2}} - q^{-\frac{1}{2}} y_j^{\frac{1}{2}} \, K^{\frac{1}{2}}\right) \; ,
\end{align}
taking into account \eqref{KX}.

We can readily see from (\ref{abcd_L=1}) that (\ref{ab_L}) holds for the case $L=1$. Next we use \eqref{mono} and \eqref{abcd} 
to write the following recurrence relations,
\begin{align}
& A_L(x) = A_{L-1}(x) A_{1}(x) + B_{L-1}(x) C_1(x)  \nonumber \\ 
& B_L(x) = A_{L-1}(x)B_1(x) + B_{L-1}(x) D_1(x) \; .
\end{align}
Thus, if \eqref{ab_L} holds for $L-1$, it follows that \eqref{ab_L} is true for arbitrary $L$.
This completes the proof.

\section{Symmetric solutions}
\label{sec:symm}

Here we demonstrate that any analytic solution of the functional equation \eqref{TPA} is a symmetric function.
Our argument closely follows the one used in \cite{Galleas_SCP}, although here we shall need only the first part 
of that argument.

As the first step of our proof we recall that the symmetric group of order $L$ is generated by any single transposition,
in addition to any cycle of length $L$. Thus it is enough to show that $\mathcal{Z}$ is invariant under cyclically 
permutations of $\lambda_1,\dots,\lambda_k$ for all $1\leq k\leq L$ in order to prove Lemma~\ref{symm_lemma}.

Next we assume $\mathcal{Z}$ is analytic and look at \eqref{TPA} in the limit~$\lambda_0\to\lambda_k$. 
From \eqref{M0MI} we see that only the coefficients $M_0$ and $M_k$ are singular as $\lambda_0\to\lambda_k$,
with residues given by
\[
\label{res}
\gen{Res}_{\lambda_0 = \lambda_k} (M_0) = - \gen{Res}_{\lambda_0 = \lambda_k} (M_k) = c\,\frac{b(2\lambda_k)}{a(2\lambda_k)} \Lambda_{\mathcal{A}} (\lambda_k) \prod_{\substack{j=1 \\ j\neq k}}^{L} \frac{a(\lambda_j - \lambda_k)}{b(\lambda_j - \lambda_k)} \frac{b(\lambda_j + \lambda_k)}{a(\lambda_j + \lambda_k)} \; .
\]
Now we use Cauchy's integral formula to integrate \eqref{TPA} with respect to $\lambda_0$ along a contour enclosing $\lambda_k$
but no other singular points. This procedure yields the following identity
\begin{align}
\label{res1}
& \gen{Res}_{\lambda_0 = \lambda_k} (M_0) \mathcal{Z}(\lambda_1,\dots,\lambda_{k-1},\lambda_k,\lambda_{k+1},\dots,\lambda_L) \nonumber \\
& \qquad - \gen{Res}_{\lambda_0 = \lambda_k} (M_k) \mathcal{Z}(\lambda_k,\lambda_1,\dots,\lambda_{k-1},\lambda_{k+1},\dots,\lambda_L) = 0 \; .
\end{align}
From \eqref{res} and \eqref{res1} we conclude that
\[
\mathcal{Z}(\lambda_1,\dots,\lambda_{k-1},\lambda_k,\lambda_{k+1},\dots,\lambda_L) =  \mathcal{Z}(\lambda_k,\lambda_1,\dots,\lambda_{k-1},\lambda_{k+1},\dots,\lambda_L) \; .
\]
This completes the proof of Lemma \ref{symm_lemma}.

\section{Special zeroes}
\label{sec:zeroes}

The strategy employed in \Secref{sec:sol} for solving Eq. (\ref{TPA}) relies on the determination of
particular zeroes of the desired solution. The location of these zeroes are stated in Lemma \ref{zeroes_lemma}
and they are as follows: $(\lambda_1 = \mu_1 -\gamma , \lambda_2=\mu_1)$ and $(\lambda_1 = \mu_1 -\gamma,\lambda_2=-\mu_1-\gamma)$.
These specialisations of variables are given in terms of the parameter $\mu_1$ but we could have considered any other 
parameter $\mu_j$ instead, as will become clear from our proof. Here we shall focus only on the first specialisation
of variables, i.e. $(\lambda_1 = \mu_1 -\gamma , \lambda_2=\mu_1)$, since the same properties can be used for showing
the second case.

We start by noticing that the coefficients $M_{i-1}$ and $M_i$ vanish for the specialisation $(\lambda_{i-1} = \mu_1 -\gamma,\lambda_i=\mu_1)$
for $2\leq i\leq L$, as can be seen from \eqref{M0MI} and \eqref{lambda}. This property is of fundamental importance for our proof.
We shall first examine the cases $L=2$ and $L=3$ for illustrative purposes before considering the general case.

\paragraph{$L=2$.} For $L=2$ the functional equation \eqref{TPA} consists of three terms and it involves the spectral parameters 
$\lambda_0,\lambda_1$ and $\lambda_2$. Upon setting $\lambda_1 = \mu_1 -\gamma$ and $\lambda_2=\mu_1$, two of the coefficients vanish
and we are left with
\[
M_0|_{1,2} \mathcal{Z}(\mu_1 -\gamma,\mu_1) = 0 \; . 
\]
Here we have written $\cdot\,|_{1,2}$ to denote the prescribed specialization of $\lambda_1$ and $\lambda_2$. The remaining coefficient
is nonzero for generic values of the inhomogeneities $\mu_j$ and parameters $\gamma$ and $h$. Thus
we can conclude that $\mathcal{Z}(\mu_1 -\gamma,\mu_1) = 0$.

\paragraph{$L=3$.} The general structure of this analysis starts to emerge at $L=3$. In that case the specialization 
$\lambda_2 = \mu_1 -\gamma$ and $\lambda_3=\mu_1$ yields the following relation,
\[
\label{L=3}
 M^0_0 \mathcal{Z}(\lambda_1,\mu_1 -\gamma,\mu_1) + M^0_1 \mathcal{Z}(\lambda_0,\mu_1 -\gamma,\mu_1) = 0 \; ,
\]
where we have written $M^0_i\coloneqq M_i|_{2,3}$.

Taking into account Remark \ref{multi} we can now produce a second equation by interchanging the variables 
$\lambda_0 \leftrightarrow \lambda_1$. For later convenience we also set $M^1_1 \coloneqq (M_0|_{2,3})|_{\lambda_0\leftrightarrow\lambda_1}$ 
and  $M^1_0 \coloneqq (M_1|_{2,3})|_{\lambda_0\leftrightarrow\lambda_1}$ such that our second equation reads
\[
\label{L=3_2}
M^1_1 \mathcal{Z}(\lambda_0,\mu_1 -\gamma,\mu_1) + M^1_0 \mathcal{Z}(\lambda_1,\mu_1 -\gamma,\mu_1) = 0 \; . 
\]
The system of equations formed by \eqref{L=3} and \eqref{L=3_2} can now be written as
\[
\label{L=3_mat}
\left( \begin{matrix} M^0_0 & M^0_1 \\ M^1_0 & M^1_1 \end{matrix} \right) \left( \begin{matrix} \mathcal{Z}(\lambda_1,\mu_1 -\gamma,\mu_1) \\ \mathcal{Z}(\lambda_0,\mu_1 -\gamma,\mu_1) \end{matrix} \right) = 0 \; ,
\]
and from the explicit expressions for the coefficients $M^j_i$ we can infer that $\det(M^j_i)\neq 0$ for arbitrary values of its variables.
Thus \eqref{L=3_mat} implies that $\mathcal{Z}(\lambda_1,\mu_1 -\gamma,\mu_1) = 0$ generically. Since $\mathcal{Z}$ is symmetric by Lemma~\ref{symm_lemma},
we have the property we wanted to show.

\paragraph{General $L$.} The general case is treated along the same lines. By setting $\lambda_{L-1} = \mu_1 -\gamma$ and $\lambda_L=\mu_1$ we obtain
the relation
\[
\label{L=general}
 M^0_0 \mathcal{Z}(\tilde{X}^{0,L-2}_0) + \sum_{i=1}^{L-2} M^0_i \mathcal{Z}(\tilde{X}^{0,L-2}_i) = 0 \; , 
\]
where $M^0_i \coloneqq M_i|_{L-1,L}$ as before. In \eqref{L=general} have further abbreviated the arguments of
$\mathcal{Z}$ as $\tilde{X}^{i,j}_l  \coloneqq  \{ \mu_1 - \gamma,\mu_1 \} \cup \{\lambda_k \mid i\leq k \leq j \} \setminus \{ \lambda_l \}$, 
which is justified by Lemma~\ref{symm_lemma}.

Now we can produce $L-2$ additional equations by switching $\lambda_0 \leftrightarrow \lambda_j$ for $2\leq j\leq L-2$
as discussed in Remark \ref{multi}. These equations can be written in the form
\[
\label{L=general_i}
M^j_0 \mathcal{Z}(\tilde{X}^{0,L-2}_0) + \sum_{i=1}^{L-2} M^j_i \mathcal{Z}(\tilde{X}^{0,L-2}_i) = 0 \; ,  \qquad 2\leq j\leq L-2 \; ,
\]
for certain coefficients $M^j_0$ and $M^j_i$. The system of equations \eqref{L=general}--\eqref{L=general_i}
can now be recasted as
\[
\label{Lmat}
\left( \begin{matrix} M^0_0 & \cdots & M^0_{L-2} \\ \vdots & \ddots & \vdots \\ M^{L-2}_0 & \cdots & M^{L-2}_{L-2} \end{matrix} \right) \left( \begin{matrix} \mathcal{Z}(\tilde{X}^{0,L-2}_0) \\ \vdots \\ \mathcal{Z}(\tilde{X}^{0,L-2}_i) \end{matrix} \right) = 0 \; .
\]
Direct inspection reveals that the matrix $M^j_i$ in \eqref{Lmat} is nonsingular for generic values of the parameters.
Thus by Lemma~\ref{symm_lemma} we can conclude that $\mathcal{Z}(\mu_1 -\gamma,\mu_1,\lambda_1,\dots,\lambda_{L-2}) = \mathcal{Z}(\tilde{X}^{0,L-2}_0) = 0$ 
generically. This completes the proof of Lemma~\ref{zeroes_lemma}.

\section{Asymptotic behavior}
\label{sec:asymp}

The functional equation (\ref{TPA}) is only able to determine the desired partition function (\ref{PF})
up to an overall multiplicative factor. In this way the full determination of $\mathcal{Z}$, as defined
in (\ref{PF}), requires we are able to compute it for a particular value of its variables. The asymptotic
behavior stated in Lemma \ref{asymp_lemma} provides us with that information and here we intend
to present its proof.

Using \eqref{abcd_L=1} and writing $x = e^{2\lambda}$, $y_i = e^{2\mu_i}$, $q = e^{\gamma}$ and $t = e^{h}$, we 
find the following asymptotic behavior as $x$ tends to infinity:
\<
\label{ASYM}
A(x) &\sim& 2^{-L} q^{\frac{L}{2}} x^{\frac{L}{2}} (K^{\frac{1}{2}})^{\otimes L} \prod_{i=1}^L y_i^{-\frac{1}{2}} \; ,  \nonumber \\
B(x) &\sim& 2^{-L} q^{\frac{L-1}{2}} (q - q^{-1})\,  x^{\frac{L-1}{2}} \sum_{j=1}^L (K^{\frac{1}{2}})^{\otimes (j-1)} \otimes X^{-} \otimes (K^{-\frac{1}{2}})^{\otimes (L-j)} \prod_{\substack{i=1 \\ i \neq j}}^L y_i^{-\frac{1}{2}} \; ,  \nonumber \\
\bar{B}(x) &\sim& 2^{-L} q^{\frac{L-1}{2}} (q - q^{-1})\, x^{\frac{L-1}{2}} \sum_{j=1}^L  (K^{-\frac{1}{2}})^{\otimes (j-1)} \otimes X^{-} \otimes (K^{\frac{1}{2}})^{\otimes (L-j)} \prod_{\substack{i=1 \\ i \neq j}}^L y_i^{\frac{1}{2}} \; ,  \nonumber \\
\bar{D}(x) &\sim& 2^{-L} q^{\frac{L}{2}} x^{\frac{L}{2}} (K^{-\frac{1}{2}})^{\otimes L} \prod_{i=1}^L y_i^{\frac{1}{2}} \; .
\>
The operators $K$ and $X^{-}$ appearing in (\ref{ASYM}) were previously defined in~\eqref{KX}.
Also, we can see from \eqref{k_entries} that $\kappa_{\pm}(x) \sim \pm 2^{-1} t^{\pm1}\, x^{\frac{1}{2}}$ as $x\to\infty$. 
This result combined with (\ref{ABCD_via_abcd}) and (\ref{ASYM}) yields the following asymptotic expansion of the operator
$\mathcal{B}$,
\[
\label{Blim}
\mathcal{B}(x) \sim \frac{q^{L-1}}{2^{2L+1}} (q - q^{-1})\, x^L \sum_{j=1}^L ( P_j^{+} + P_j^{-} ) \; ,
\]
where we have set
\[
P_j^{\pm}  \coloneqq  \pm (t\, y_j^{\frac{1}{2}})^{\pm1} \, \gen{id}^{\otimes(j-1)} \otimes X^{-} \otimes ( K^{\pm 1} )^{\otimes(L-j)} \; .
\]
From \eqref{Uqsl2} it follows that the operators $P_j^{\pm}$ satisfy the following commutation rules:
\begin{align}
\label{Ppm}
P_i^{\pm} P_j^{\pm} &= q^{\mp2}  P_j^{\pm} P_i^{\pm} \; , \qquad P_i^{\pm} P_j^{\mp} = q^{\mp 2}  P_j^{\mp} P_i^{\pm} \; , && \text{for } \; i<j \; , \nonumber \\
P_i^{s} P_i^{s'} & =0  \; && \text{for } \;  s,s' \in \{\pm\} \; . 
\end{align}

The behavior of (\ref{PF}) in the limit $x_i \to \infty$ for $1 \leq i \leq L$ can now be computed using (\ref{Blim}).
To this end it is convenient to introduce the operators
\[
\label{QJ}
Q_j^{(n)}  \coloneqq  P_j^{+} q^{-2n} + P_j^{-} q^{2n}
\]
such that
\[
\label{Zlim}
\bar{\mathcal{Z}} \sim \frac{q^{L(L-1)}}{2^{L(2L+1)}} (q - q^{-1})^L \left(\prod_{i=1}^L x_i^{2L}\right) \sum_{j_1 = 1}^L \dots \sum_{j_L = 1}^L \mathop{\overrightarrow\prod}\limits_{1 \le k \le L } Q_{j_k}^{(0)} \; .
\]
The operators $Q_j^{(n)}$ as defined in (\ref{QJ}) satisfy the following commutation relations,
\<
\label{QIQJ}
Q_i^{(n)} Q_j^{(0)} &=& Q_j^{(0)} Q_i^{(n+1)} \qquad \text{for } \; i<j \; , \nonumber \\
Q_i^{(m)} Q_i^{(n)} &=& 0 \; ,   
\>
as a direct consequence of (\ref{Ppm}). Now due to the last relation of (\ref{QIQJ}), the summation in 
the right-hand side of (\ref{Zlim}) reduces to
\<
\label{QLN}
\sum_{j_1 = 1}^L \dots \sum_{j_L = 1}^L \mathop{\overrightarrow\prod}\limits_{1 \le k \le L } Q_{j_k}^{(0)} = \sum_{\sigma \in \mathcal{S}_L} \mathop{\overrightarrow\prod}\limits_{1 \le i \le L } Q_{\sigma(i)}^{(0)} \; ,
\>
where $\mathcal{S}_L$ is the symmetric group of order $L$. The relation (\ref{QLN}) can be further simplified with the help of the 
first relation in (\ref{QIQJ}). In this way we are left with
\[
\label{QQ}
\sum_{j_1 = 1}^L \dots \sum_{j_L = 1}^L \mathop{\overrightarrow\prod}\limits_{1 \le k \le L } Q_{j_k}^{(0)} = \mathop{\overrightarrow\prod}\limits_{0 \le n \le L-1 } \left( \sum_{m=0}^n Q_{L-n}^{(m)} \right) \; .
\]
Next we notice that 
\[
\sum_{m=0}^n Q_{L-n}^{(m)} = P_{L-n}^{+} \Delta_n^{+} + P_{L-n}^{-} \Delta_n^{-}
\]
with $\Delta_n^{\pm}  \coloneqq  \sum_{m=0}^{n} q^{\pm 2m}$. Thus we can compute the matrix element $\bra{\bar{0}} \mathcal{N} \ket{0}$
with $\mathcal{N}$ given by (\ref{QQ}) straightforwardly. By doing so we obtain,
\<
\label{almost}
\bra{\bar{0}} \mathop{\overrightarrow\prod}\limits_{0 \le n \le L-1 } \left( \sum_{m=0}^n Q_{L-n}^{(m)} \right) \ket{0} = \prod_{n=0}^{L-1} ( t\, y_{L-n}^{-\frac{1}{2}} q^n \Delta_n^{+} - t^{-1} y_{L-n}^{\frac{1}{2}} q^{-n} \Delta_n^{-} ) \; .
\>
The expression (\ref{almost}) can be further simplified by noticing that $q^n \Delta_n^{+} = q^{-n} \Delta_n^{-}$. This reduces the right-hand side of (\ref{almost}) to $q^{- \frac{L(L-1)}{2}} [ L! ]_{q^2} \prod_{i=1}^L ( t\, y_i^{-\frac{1}{2}} - t^{-1} y_i^{\frac{1}{2}} )$. 
Gathering our results we arrive at formula (\ref{asymp_form}).

\section{Solution for $L=1$}
\label{sec:L1}

The functional equation (\ref{TPA}) for $L=1$ reads $M_0 \mathcal{Z}(\lambda_1) + M_1 \mathcal{Z}(\lambda_0) = 0$,
which simplifies to 
\[
\sinh{(2 \lambda_0)} \mathcal{Z}(\lambda_1) - \sinh{(2 \lambda_1)} \mathcal{Z}(\lambda_0) = 0 \; ,
\]
upon the use of the explicit expressions for $M_0$ and $M_1$ given in (\ref{M0MI}). Thus we readily
find the separation of variables
\[
\frac{\mathcal{Z}(\lambda_0)}{\sinh{(2 \lambda_0)}} = \frac{\mathcal{Z}(\lambda_1)}{\sinh{(2 \lambda_1)}}  \; ,
\]
leading to the solution
\[
\label{sol1}
\mathcal{Z}(\lambda) = k \sinh{(2 \lambda)} \; .
\]
Here $k$ is a constant that is fixed to be $k = \sinh{(\gamma)} \sinh{(h - \mu_1)}$ by the asymptotic behavior discussed in \Appref{sec:asymp}.
The solution (\ref{sol1}) can still be recasted as the following contour integral,
\[
\label{int1}
\mathcal{Z}(\lambda) = \oint \frac{\dd w_1}{2 \ii \pi} \frac{H (w_1)}{\sinh{(w_1 - \lambda)}} \; ,
\]
where the function $H$ is given by
\<
\label{H1}
H(w_1) &=& c \; \frac{b(h - \mu_1)}{b(h + \mu_1)} \frac{b(2 w_1)}{a(2 w_1)} \left\{ \frac{b(w_1 + h)}{a(w_1 - \mu_1)} a(w_1 - \mu_1) a(w_1 + \mu_1) \right. \nonumber \\
&& \qquad \qquad \qquad \qquad \qquad \left.   - \;  \frac{a(w_1 - h)}{b(w_1 + \mu_1)} b(w_1 - \mu_1) b(w_1 + \mu_1) \right\} \; . \nonumber \\ 
\>
Here we have already used the explicit form of the constant~$k$. Also, we have used some redundancies in formula (\ref{H1}) in order to make 
the connection with the relation (\ref{HH}) more explicit.

\section{Reduction of order}
\label{sec:RED}

In \Secref{sec:PDE} we have unveiled a set of linear partial differential equations underlying the functional
relation (\ref{TPA}). Those equations are given formally by (\ref{omg_k}) and the explicit construction of the 
set of differential operators $\{ \gen{\Omega}_k  \}$ was also discussed in \Secref{sec:PDE}. In particular, we
found a compact expression for the operator $\gen{\Omega}_{2L}$ which is given by (\ref{O2L}) and (\ref{UY}). From (\ref{O2L}) we 
see that Eq. $\gen{\Omega}_{2L} \; \bar{\mathcal{Z}}(X^{1,L}) = 0$ is of order $2L$ and can be recasted as a
system of first-order equations. The resulting system of equations is described in the following lemma.

\begin{lemma} \label{SOE}
Let $\psi^{(0)} \coloneqq \bar{\mathcal{Z}} (x_1 , \dots , x_L)$ and let $\psi_i^{(k)}=\psi_i^{(k)}(x_1 , \dots , x_L)$
for $1 \leq i \leq L$ and $1 \leq k \leq 2L-1$ be multivariate functions.
Then the differential equation $\gen{\Omega}_{2L} \; \bar{\mathcal{Z}} = 0$ is equivalent to the following 
system of equations,
\<
\label{soe}
\mathcal{U} \psi^{(0)}  + \sum_{i=1}^L \mathcal{Y}_i \; \partial_i \psi_i^{(2L-1)} &=& 0 \; , \nonumber \\
\psi_i^{(1)} - \partial_i \psi_0 &=& 0  \qquad 1 \leq i \leq L \; , \nonumber \\
\psi_i^{(k)} - \partial_i \psi_i^{(k-1)} &=& 0 \qquad 1 \leq i \leq L \; , \; 2 \leq k \leq 2L-1 \; ,  
\>
where $\partial_i \coloneqq \frac{\partial}{\partial x_i}$.
\end{lemma}
\begin{proof}
The verification is straightforward.
\end{proof}

\paragraph{Matricial form.} In order to further enhance the structure of (\ref{soe}), we finally
recast our system of first-order equations as a matrix equation. For that we define the ($(2L-1)L+1$)-component 
vector 
\[
\psi= \psi(x_1,\dots,x_L) \coloneqq \left( \begin{matrix}
\psi^{(0)} \cr
\psi^{(1)}_1 \cr
\vdots \cr
\psi^{(1)}_L \cr
\vdots \cr
\psi^{(2L-1)}_1 \cr
\vdots \cr
\psi^{(2L-1)}_L \end{matrix} \right) \; .
\]
In this way the system of equations (\ref{soe}) is equivalent to $\mathcal{H} \psi = 0$ where
\[
\label{HP}
\mathcal{H} \coloneqq \left( \begin{array}{c|ccccc}
\mathcal{U} &  &  &  & \vec{\omega} \\[0ex] \hline \\[-2ex]
\vec{\nabla} & -\mathbbm{1} &  & & \\
 & \mathfrak{D} & -\mathbbm{1} & &  \\
 & & \ddots & \ddots &  \\
 & &  & \mathfrak{D} & -\mathbbm{1} \\
 \end{array} \right) \; .
\]
In (\ref{HP}) the null entries are suppressed while $\mathcal{U}$ is the function defined in (\ref{UY}).
Moreover, the first-order differential operators are given by
\<
\vec{\omega} &\coloneqq & (\mathcal{Y}_1 \partial_1 , \dots , \mathcal{Y}_L \partial_L) \; , \nonumber \\
\mathfrak{D} &\coloneqq & \mbox{diag}(\partial_1 , \dots , \partial_L) \; , \nonumber \\
\vec{\nabla} &\coloneqq &  \left( \begin{matrix} \partial_1 \cr  \vdots \cr \partial_L \end{matrix} \right) \; ,
\>
with functions $\mathcal{Y}_i$ defined in (\ref{UY}) and $\mathbbm{1}$ is the $L \times L$ identity matrix.


\bibliographystyle{hunsrt}
\bibliography{references}

\begin{thebibliography}{10}

\bibitem{Korepin_1982}
V.~E. Korepin.
\newblock {Calculation of norms of Bethe wave functions}.
\newblock {\em Commun. Math. Phys.}, 86:391--418, 1982.

\bibitem{Izergin_1987}
A.~G. Izergin.
\newblock Statistical sum of the six-vertex model in a finite lattice.
\newblock {\em Sov. Phys. Dokl.}, 32:878, 1987.

\bibitem{Slavnov_1989}
N.~A. Slavnov.
\newblock {Calculation of scalar products of wave functions and form factors in
  the framework of the algebraic Bethe ansatz}.
\newblock {\em {Theor. Math. Phys.}}, {79}({2}):{502--508}, {1989}.

\bibitem{Gaudin_book}
M.~Gaudin.
\newblock {\em La fonction d'onde de Bethe}.
\newblock Masson, Paris, 1983.

\bibitem{Korepin_book}
V.~E. Korepin, N.~M. Bogoliubov, and A.~G. Izergin.
\newblock {\em Quantum inverse scattering method and correlation functions}.
\newblock Cambridge University Press, 1993.

\bibitem{Its_1990}
A.~R. Its, A.~G. Izergin, V.~E. Korepin, and N.~A. Slavnov.
\newblock {Differential equations for quantum correlation functions}.
\newblock {\em {Int. J. Mod. Phys. B}}, {4}:{1003--1037}, {1990}.

\bibitem{Knizhnik_1984}
V.~G. Knizhnik and A.~B. Zamolodchikov.
\newblock {Current algebra and Wess-Zumino model in two dimensions}.
\newblock {\em {Nucl. Phys. B}}, {247}({1}):{83--103}, {1984}.

\bibitem{Belliard1}
S.~Belliard, S.~Pakuliak, E.~Ragoucy, and N.~A. Slavnov.
\newblock {Highest coefficient of scalar products in $SU(3)$-invariant
  integrable models}.
\newblock {\em J. Stat. Mech.}, 09:P09003, 2012, {arXiv:1206.4931 [math-ph]}.

\bibitem{Belliard2}
S.~Belliard, S.~Pakuliak, E.~Ragoucy, and N.~A. Slavnov.
\newblock {The algebraic Bethe ansatz for scalar products in $SU(3)$-invariant
  integrable models}.
\newblock {\em J. Stat. Mech.}, 10:P10017, 2012, {arXiv:1207.0956 [math-ph]}.

\bibitem{Belliard3}
S.~Belliard, S.~Pakuliak, E.~Ragoucy, and N.~A. Slavnov.
\newblock {Bethe vectors of $GL(3)$-invariant integrable models}.
\newblock {\em J. Stat. Mech.}, 02:P02020, 2013, {arXiv:1210.0768 [math-ph]}.

\bibitem{Belliard4}
S.~Belliard, S.~Pakuliak, E.~Ragoucy, and N.~A. Slavnov.
\newblock {Form factors in $SU(3)$-invariant integrable models}.
\newblock {\em J. Stat. Mech.}, 09:P04033, 2013, {arXiv:1211.3968 [math-ph]}.

\bibitem{Kitanine_1999}
N.~Kitanine, J.~M. Maillet, and V.~Terras.
\newblock {Form factors of the $XXZ$ Heisenberg spin-$1/2$ finite chain}.
\newblock {\em {Nucl. Phys. B}}, {554}:{647--678}, {1999},
  arXiv:math-ph/9807020.

\bibitem{Kitanine_2000}
N.~Kitanine, J.~M. Maillet, and V.~Terras.
\newblock {Correlation functions of the $XXZ$ Heisenberg spin-$1/2$ chain in a
  magnetic field}.
\newblock {\em {Nucl. Phys. B}}, {567}:{554--582}, {2000},
  arXiv:math-ph/9907019.

\bibitem{Kitanine_2002}
N.~Kitanine, J.~M. Maillet, N.~A. Slavnov, and V.~Terras.
\newblock {Spin-spin correlation functions of the $XXZ$-$1/2$ Heisenberg chain
  in a magnetic field}.
\newblock {\em {Nucl. Phys. B}}, {641}:{487--518}, {2002},
  arXiv:hep-th/0201045.

\bibitem{Kitanine_2005}
N.~Kitanine, J.~M. Maillet, N.~A. Slavnov, and V.~Terras.
\newblock {On the spin-spin correlation functions of the $XXZ$ spin-$1/2$
  infinite chain}.
\newblock {\em {J. Phys. A: Math. Gen.}}, {38}:{7441--7460}, {2005},
  arXiv:hep-th/0407223.

\bibitem{Terras_2013a}
D.~Levy-Bencheton and V.~Terras.
\newblock {An algebraic Bethe ansatz approach to form factors and correlation
  functions of the cyclic eight-vertex solid-on-solid model}.
\newblock {\em J. Stat. Mech.}, 04:P04015, 2013, arXiv:1212.0246 [math-ph].

\bibitem{Terras_2013b}
D.~Levy-Bencheton and V.~Terras.
\newblock {Spontaneous staggered polarizations of the cyclic solid-on-solid
  model from the algebraic Bethe Ansatz}.
\newblock {\em {J. Stat. Mech.}}, {10}:{P10012}, {2013}, arXiv:1304.7814
  [math-ph].

\bibitem{Niccoli_2013a}
G.~Niccoli.
\newblock {Form factors and complete spectrum of $XXX$ antiperiodic higher spin
  chains by quantum separation of variables}.
\newblock {\em {J. Math. Phys.}}, {54}({5}), {2013}, arXiv:1206.2418 [math-ph].

\bibitem{Niccoli_2013b}
G.~Niccoli.
\newblock {Antiperiodic spin-1/2 XXZ quantum chains by separation of variables:
  Complete spectrum and form factors}.
\newblock {\em {Nucl. Phys. B}}, {870}({2}):{397--420}, {2013}, arXiv:1205.4537
  [math-ph].

\bibitem{Jimbo_1993}
B.~Davies, O.~Foda, M.~Jimbo, T.~Miwa, and A.~Nakayashiki.
\newblock {Diagonalization of the $XXZ$ Hamiltonian by vertex operators}.
\newblock {\em Commun. Math. Phys.}, 151:89--153, 1993, arXiv:hep-th/9204064.

\bibitem{Kedem_1995a}
M.~Jimbo, R.~Kedem, T.~Kojima, H.~Konno, and T.~Miwa.
\newblock {$XXZ$ chain with a boundary}.
\newblock {\em {Nucl. Phys. B}}, {441}({3}):{437--470}, {1995},
  arXiv:hep-th/9411112.

\bibitem{Baseilhac_2013}
P.~Baseilhac and S.~Belliard.
\newblock {The half-infinite $XXZ$ chain in Onsager's approach}.
\newblock {\em {Nucl. Phys. B}}, {873}({3}):{550--584}, {2013}, arXiv:1211.6304
  [math-ph].

\bibitem{Jimbo_1992}
M.~Jimbo, K.~Miki, T.~Miwa, and A.~Nakayashiki.
\newblock {Correlation functions of the $XXZ$ model for $\Delta < -1$}.
\newblock {\em {Phys. Lett. A}}, {168}({4}):{256--263}, {1992},
  arXiv:hep-th/9205055.

\bibitem{Jimbo_1993a}
M.~Jimbo, T.~Miwa, and A.~Nakayashiki.
\newblock {Difference equations for the correlation functions of the
  eight-vertex model}.
\newblock {\em {J. Phys. A: Math. Gen.}}, {26}({9}):{2199--2209}, {1993},
  arXiv:hep-th/9211066.

\bibitem{Jimbo_1996}
M.~Jimbo and T.~Miwa.
\newblock {Quantum KZ equation with $|q|=1$ and correlation functions of the
  $XXZ$ model in the gapless regime}.
\newblock {\em {J. Phys. A: Math. Gen.}}, {29}({12}):{2923--2958}, {1996},
  arXiv:hep-th/9601135.

\bibitem{Reshetikhin_1992}
I.~B. Frenkel and N.~Yu. Reshetikhin.
\newblock {Quantum affine algebras and holonomic difference equations}.
\newblock {\em {Commun. Math. Phys.}}, {146}({1}):{1--60}, {1992}.

\bibitem{Galleas_proc}
W.~Galleas.
\newblock {Functional relations and the Yang-Baxter algebra}.
\newblock {\em {Journal of Physics: Conference Series}}, {474}:{012020},
  {2013}, {arXiv:1312.6816 [math-ph]}.

\bibitem{Galleas_2011}
W.~Galleas.
\newblock A new representation for the partition function of the six-vertex
  model with domain wall boundaries.
\newblock {\em J. Stat. Mech.}, 01:P01013, 2011, {arXiv:1010.5059 [math-ph]}.

\bibitem{Galleas_2008}
W.~Galleas.
\newblock {Functional relations from the Yang-Baxter algebra: Eigenvalues of
  the $XXZ$ model with non-diagonal twisted and open boundary conditions}.
\newblock {\em {Nucl. Phys. B}}, {790}({3}):{524--542}, {2008},
  {arXiv:0708.0009 [nlin.SI]}.

\bibitem{Galleas_2010}
W.~Galleas.
\newblock Functional relations for the six-vertex model with domain wall
  boundary conditions.
\newblock {\em J. Stat. Mech.}, 06:P06008, 2010, {arXiv:1002.1623 [math-ph]}.

\bibitem{Galleas_2012}
W.~Galleas.
\newblock {Multiple integral representation for the trigonometric SOS model
  with domain wall boundaries}.
\newblock {\em {Nucl. Phys. B}}, {858}({1}):{117--141}, {2012},
  {arXiv:1111.6683 [math-ph]}.

\bibitem{Galleas_2013}
W.~Galleas.
\newblock {Refined functional relations for the elliptic SOS model}.
\newblock {\em {Nucl. Phys. B}}, {867}:{855--871}, {2013}, {arXiv:1207.5282
  [math-ph]}.

\bibitem{Galleas_SCP}
W.~Galleas.
\newblock {Scalar product of Bethe vectors from functional equations}.
\newblock {\em {Comm. Math. Phys.}}, {329}({1}):{141--167}, {2014},
  {arXiv:1211.7342 [math-ph]}.

\bibitem{Sklyanin_1988}
E.~K. Sklyanin.
\newblock {Boundary conditions for integrable quantum systems}.
\newblock {\em {J. Phys. A: Math. Gen.}}, {21}({10}):{2375--2389}, {1988}.

\bibitem{Kitanine_2007}
N.~Kitanine, K.~K. Kozlowski, J.~M. Maillet, G.~Niccoli, N.~A. Slavnov, and
  V.~Terras.
\newblock {Correlation functions of the open $XXZ$ chain: I}.
\newblock {\em {J. Stat. Mech.}}, {10}:{P10009}, {2007}, {arXiv:0707.1995
  [hep-th]}.

\bibitem{Kitanine_2008}
N.~Kitanine, K.~K. Kozlowski, J.~M. Maillet, G.~Niccoli, N.~A. Slavnov, and
  V.~Terras.
\newblock {Correlation functions of the open $XXZ$ chain: II}.
\newblock {\em {J. Stat. Mech.}}, {07}:{P07010}, {2008}, {arXiv:0803.3305
  [hep-th]}.

\bibitem{Kedem_1995b}
M.~Jimbo, R.~Kedem, H.~Konno, T.~Miwa, and R.~Weston.
\newblock {Difference equations in spin chains with a boundary}.
\newblock {\em {Nucl. Phys. B}}, {448}({3}):{429--456}, {1995},
  {arXiv:hep-th/9502060}.

\bibitem{Tsuchiya_1998}
O.~Tsuchiya.
\newblock {Determinant formula for the six-vertex model with reflecting end}.
\newblock {\em {J. Math. Phys.}}, {39}({11}):{5946--5951}, {1998},
  {arXiv:solv-int/9804010}.

\bibitem{Goehmann_2005}
F.~G\"ohmann, M.~Bortz, and H.~Frahm.
\newblock {Surface free energy for systems with integrable boundary
  conditions}.
\newblock {\em {J. Phys. A: Math. Gen.}}, {38}({50}):{10879--10891}, {2005},
  {arXiv:cond-mat/0508377}.

\bibitem{Kozlowski_2012}
K.~K. Kozlowski and B.~Pozsgay.
\newblock {Surface free energy of the open $XXZ$ spin-$1/2$ chain}.
\newblock {\em {J. Stat. Mech.}}, {05}:{P05021}, {2012}, {arXiv:1201.5884
  [nonlin.SI]}.

\bibitem{Gaudin_1971}
M.~Gaudin.
\newblock {Boundary energy of a Bose gas in one dimension}.
\newblock {\em {Phys. Rev. A}}, {4}({1}):{386}, {1971}.

\bibitem{Alcaraz_1987}
F.~C. Alcaraz, M.~N. Barber, M.~T. Batchelor, R.~J. Baxter, and G.~R.~W.
  Quispel.
\newblock {Surface exponents of the quantum $XXZ$, Ashkin-Teller and Potts
  models}.
\newblock {\em {J. Phys. A: Math. Gen.}}, {20}({18}):{6397--6409}, {1987}.

\bibitem{Korepin_Justin_2000}
V.~Korepin and P.~Zinn-Justin.
\newblock {Thermodynamic limit of the six-vertex model with domain wall
  boundary conditions}.
\newblock {\em {J. Phys. A: Math. Gen.}}, {33}({40}):{7053--7066}, {2000},
  {arXiv:cond-mat/0004250}.

\bibitem{Cardy_1986}
J.~L. Cardy.
\newblock {Effect of boundary conditions on the operator content of
  two-dimensional conformally invariant theories}.
\newblock {\em {Nucl. Phys. B}}, {275}({2}):{200--218}, {1986}.

\bibitem{Cherednik_1984}
I.~V. Cherednik.
\newblock {Factorizing particles on a half-line and root systems}.
\newblock {\em {Theor. Math. Phys.}}, {61}({1}):{977--983}, {1984}.

\bibitem{Baxter_book}
R.~J. Baxter.
\newblock {\em {Exactly Solved Models in Statistical Mechanics}}.
\newblock Dover Publications, Inc., Mineola, New York, 2007.

\bibitem{Galleas_2014}
W.~Galleas.
\newblock {Partial differential equations from integrable vertex models}.
\newblock {2014}, {arXiv:1403.0425 [math-ph]}.

\bibitem{Filali_2010}
G.~Filali and N.~Kitanine.
\newblock {The partition function of the trigonometric SOS model with a
  reflecting end}.
\newblock {\em J. Stat. Mech.}, 06:L06001, 2010, arXiv:1004.1015 [math-ph].

\bibitem{Filali_2011}
G.~Filali.
\newblock {Elliptic dynamical reflection algebra and partition function of SOS
  model with reflecting end}.
\newblock {\em {J. Geom. Phys.}}, {61}({10}):{1789--1796}, {2011},
  arXiv:1012.0516 [math-ph].

\end{thebibliography}

\end{document}